\documentclass[11pt,a4paper]{article}
\usepackage{mathtools}
\usepackage{authblk} 
\usepackage{algpseudocode,algorithmicx,algorithm}

\usepackage{mathrsfs}
\usepackage{latexsym,bm}
\usepackage{amsmath,amsfonts,amssymb,amsthm}
\usepackage{extarrows}

\usepackage{graphicx,subfigure,epstopdf,float}
\usepackage{placeins}
\usepackage{enumerate,cases,multirow}
\usepackage{makecell}
\usepackage{caption}
\usepackage{xcolor}

\usepackage{longtable,colortbl,arydshln,threeparttable}
\definecolor{mygray}{gray}{.9}

\usepackage{indentfirst}
\usepackage[top=25mm,bottom=20mm,left=25mm,right=20mm]{geometry}

\usepackage{cite}

\usepackage{listings}

\usepackage{makeidx}        
\usepackage{booktabs}
\usepackage[bookmarks,bookmarksnumbered,colorlinks,citecolor=blue,linkcolor=red,hyperindex,linktocpage=true]{hyperref}

\newcommand{\ket}[1]{| #1 \rangle} 
\newcommand{\bra}[1]{\langle #1 |} 

\newcommand{\ketbra}[2]{ | #1 \rangle \langle #2 |}
\newcommand{\bb}{\boldsymbol}

\def \d {\mathrm{d}}
\def \e {\mathrm{e}}
\def \i {\mathrm{i}}

\DeclareMathOperator{\diag}{diag}
\DeclareMathOperator{\erf}{erf}
\DeclareMathOperator{\erfc}{erfc}

\newcounter{parentalgorithm}

\makeatother

\newtheorem{theorem}{Theorem}[section]
\newtheorem{lemma}{Lemma}[section]
\newtheorem{proposition}{Proposition}[section]
\newtheorem{corollary}{Corollary}[section]

\theoremstyle{remark}
\newtheorem{remark}{\bf Remark}[section]

\numberwithin{equation}{section}

\hypersetup{CJKbookmarks=true}
\usepackage[qm]{qcircuit}

\begin{document}

\title{\bfseries A unifying framework for quantum algorithms for time-dependent non-unitary dynamics}

	\author[1,2,3]{Xiaojing Dong\thanks{dongxiaojing99@xtu.edu.cn}}
	\author[4,5]{Chuwen Ma\thanks{cwma@math.ecnu.edu.cn}}
    \author[1,2,3]{Yizhe Peng\thanks{pengyizhe@smail.xtu.edu.cn}}
	\author[1,2,3]{Yue Yu\thanks{terenceyuyue@xtu.edu.cn} }

\affil[1]{School of Mathematics and Computational Science, Xiangtan University, Xiangtan, Hunan 411105, China}
\affil[2]{Hunan Research Center of the Basic Discipline Fundamental Algorithmic Theory and Novel Computational Methods, Xiangtan, Hunan 411105, China}
\affil[3]{National Center for Applied Mathematics in Hunan, Xiangtan, Hunan 411105, China}
\affil[4]{School of Mathematical Sciences, Key Laboratory of MEA, Ministry of Education, East China Normal University, Shanghai 200241, China
	}
\affil[5]{
		Shanghai Key Laboratory of PMMP, East China Normal University, Shanghai 200241, China}

	\maketitle

\begin{abstract}
We present an improved autonomization method for quantum simulation of time-dependent homogeneous dissipative linear systems, combining a clock-variable reformulation with Schr\"odingerization to obtain a time-independent Hamiltonian system. To control both discretization error and recovery probability, we construct the clock profile from a compactly supported window function and a normalized Dirichlet kernel: the former imposes endpoint regularity, while the latter concentrates the sampled mass and keeps the discrete normalization bounded. For the subsequent Schr\"odingerization step, we use an exactly periodic version of an error-function initial profile, whose explicit Fourier coefficients allow direct error estimates and simplify the analysis of smooth initialization. By combining Fourier projection with weighted recovery, we retain a single logarithmic precision factor in the full algorithm. Under finite-order derivative bounds, sampled matrix access, and exact state-preparation access, we prepare the normalized solution at time $T$ to accuracy $\varepsilon$ with success probability $\Theta(1)$ using $\mathcal O\!\bigl(g\alpha_AT\log(g\alpha_AT/\varepsilon)\bigr)$ matrix queries and $\mathcal O(g)$ queries to each state-preparation oracle, where $g=\|\bb{x}_0\|/\|\bb{x}(T)\|$ and $\alpha_A$ is the matrix-oracle normalization. We illustrate the construction and its recovery probabilities through a numerical experiment on a time-dependent two-cavity system.
\end{abstract}

\medskip \noindent\textbf{Key words.} quantum algorithms, autonomization, Schr\"odingerization, Hamiltonian simulation, Fourier collocation



\section{Introduction}
Quantum algorithms provide a potential approach to high-dimensional linear dynamics, for which classical algorithms often suffer from the curse of dimensionality. Quantum computers are naturally suited for simulating unitary dynamics, and Hamiltonian simulation is a central algorithmic task in quantum computation~\cite{Berry-Childs-Kothari-2015,Low2019Interaction,BerryChilds2020TimeHamiltonian}. However, many problems in scientific computing give rise to non-unitary dynamics, for instance through dissipation, source terms, boundary conditions, and spatial discretizations of time-dependent partial differential equations (PDEs)~\cite{Cao2013Poisson,Berry2014Highorder,qFEM-2016,Costa2019Wave,Engel2019qVlasov,Childs-Liu-2020,Linden2020heat,Childs2021high,JLY2022multiscale,JLY22nonlinear,JLLY23ABC,JLY22SchrLong,JLY22SchrShort,JLLY2024boundary,JinLiu2022nonlinear}. In general, these non-unitary dynamics cannot be implemented directly as unitary circuits and must instead be encoded in a form compatible with quantum linear-algebra primitives or Hamiltonian simulation.

In this paper, we consider the homogeneous time-dependent, or non-autonomous, linear system
\begin{equation}
\label{ODElinear}
\frac{\d\bb{x}(t)}{\d t}=A(t)\bb{x}(t), \qquad 0\le t\le T, \qquad \bb{x}(0)=\bb{x}_0,
\end{equation}
where $A(t)\in\mathbb C^{N\times N}$, $N=2^n$, and $A(t)+A(t)^\dagger\preceq O$ on $[0,T]$. Our task is to prepare a quantum state proportional to $\bb{x}(T)$, within a prescribed precision $\varepsilon$ and with success probability $\Theta(1)$, under the assumptions in Section~\ref{sec:problem}. System~\eqref{ODElinear} contains time-dependent Hamiltonian simulation as a special case: when $A(t)$ is anti-Hermitian, the evolution is unitary and can be simulated using time-dependent Hamiltonian simulation methods~\cite{Berry-Childs-Kothari-2015,Low2019Interaction,BerryChilds2020TimeHamiltonian}. For a general dissipative linear differential equation, non-unitarity must also be accommodated. Existing quantum algorithms for such dynamics can be discussed in two settings, depending on whether the coefficients are time-independent or time-dependent.

For time-independent, or autonomous, systems with $A(t)\equiv A$, many quantum ODE solvers have been developed. In one approach, the time variable is discretized, and the resulting linear system is solved using a quantum linear system algorithm (QLSA)~\cite{Berry2014Highorder,BerryChilds2017ODE,Childs-Liu-2020,JLY2022multiscale,WL24,KroviODE,Dong2025Pade}. A representative result is the algorithm introduced in~\cite{BerryChilds2017ODE}, referred to here as the BCOW algorithm, which achieves an exponential improvement in the precision dependence through a truncated Taylor-series construction. Alternatively, the non-unitary evolution can be represented as a linear combination of unitary evolutions and implemented using linear combination of unitaries (LCU). This includes linear combination of Hamiltonian simulation (LCHS)~\cite{ALL2023LCH,ACL2023LCH2} and its optimal-scaling construction~\cite{low2025optimal}, the transmutation method based on the Kannai transform~\cite{JinMaZuazua2026Transmutation}, and the linear-combination representation for non-Hermitian matrix functions~\cite{wang2025quantum}. In a third approach, non-unitary dynamics are embedded into a larger unitary or Hamiltonian system, as in Schr\"odingerization~\cite{JLY22SchrShort,JLY22SchrLong,JLM24SchrInhom,JLMPY2025schr} and moment-matching dilation~\cite{Li25}. These approaches provide different routes to encoding time-independent non-unitary linear dynamics into quantum algorithms.

For time-dependent (non-autonomous) systems, the design of quantum algorithms becomes more challenging. The primary difficulty stems from the time-ordering requirement: since $A(t_i)$ and $A(t_j)$ generally do not commute for $t_i\ne t_j$, the propagator cannot in general be expressed as the ordinary exponential $\exp(\int_s^t A(\tau)\,\d\tau)$. Instead, it is represented by the time-ordered exponential $U_A(t,s)=\exp_{\mathcal T}(\int_s^t A(\tau)\,\d\tau)$. Existing approaches address this dependence using truncated Dyson-series or time-marching constructions~\cite{BerryCosta2024timeODE,Lin2023timeODE}, or time-dependent Hamiltonian simulation within an LCHS representation~\cite{ACL2023LCH2}. Compared with the autonomous setting, algorithms must therefore account for time ordering as well as the non-unitarity of the evolution.

\begin{figure}[!htb]
\centering
\includegraphics[width=0.7\textwidth]{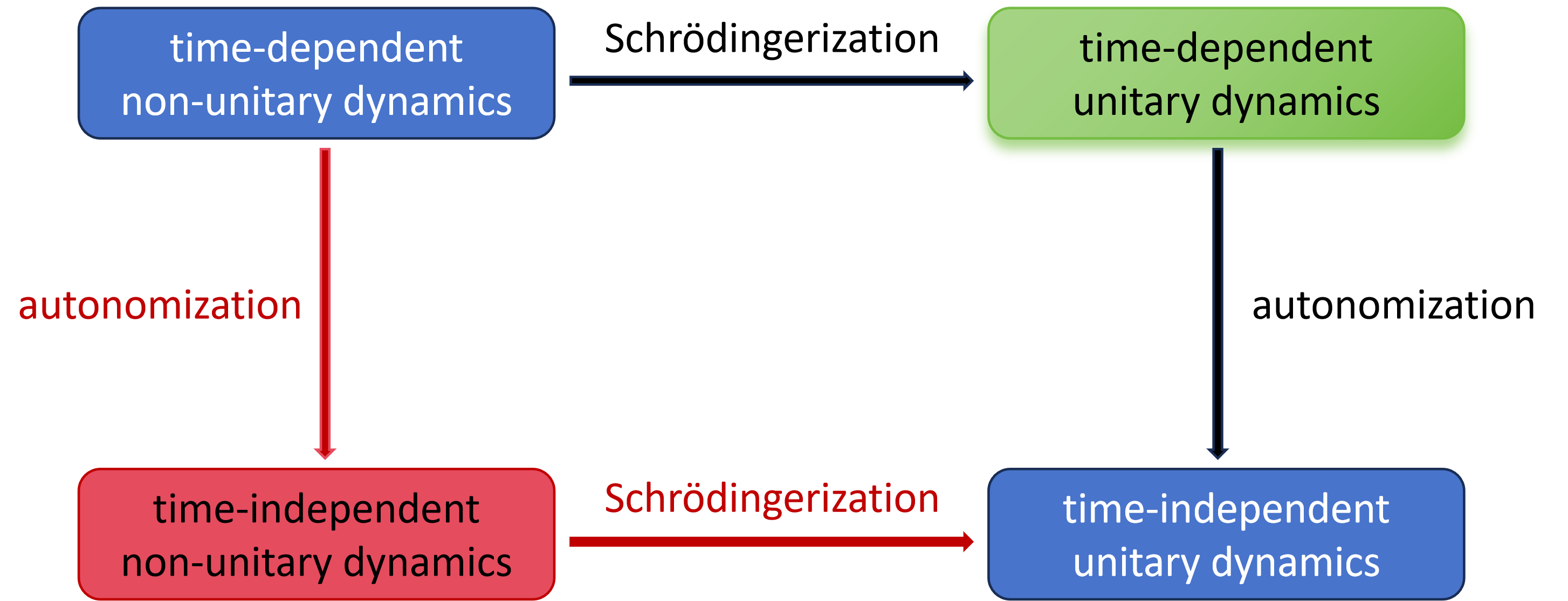}
\caption{Two routes from time-dependent non-unitary dynamics to time-independent unitary dynamics. Along the upper path, Schr\"odingerization is followed by autonomization; along the lower path, the order is reversed.}
\label{fig:autonomization-relations}
\end{figure}

This motivates us to seek a time-independent reformulation of the original non-autonomous system, a procedure we refer to as \emph{autonomization}. We consider the two routes in Fig.~\ref{fig:autonomization-relations}.
\begin{itemize}
\item For time-dependent Hamiltonian systems, Ref.~\cite{CJL23TimeSchr} proposed a formulation that transforms non-autonomous Hamiltonian systems into autonomous ones in one additional dimension. Combined with Schr\"odingerization, this dilation technique also applies to time-dependent non-unitary dynamics, as represented by the upper path: Schr\"odingerization first yields time-dependent unitary dynamics, which are then converted into time-independent unitary dynamics through autonomization.
\item A natural question then arises: can we instead take the lower path and first autonomize the non-unitary system itself? This alternative offers flexibility in the choice of solver, because we can replace the subsequent Schr\"odingerization step with a quantum ODE solver designed for time-independent non-unitary dynamics, provided that its stability and oracle assumptions are satisfied.
\end{itemize}

In this work, we establish the feasibility of the second route for the dissipative system~\eqref{ODElinear}. We show that the autonomization technique of Ref.~\cite{CJL23TimeSchr} remains valid in this homogeneous non-unitary setting. Specifically, we embed the time-dependent system into an autonomous convection equation in an auxiliary variable $s$, and recover the original solution by evaluating the enlarged solution at $s=T$, where $T$ is the total evolution time. By choosing an appropriate initial function, we can impose periodic boundary conditions in the $s$-direction and subsequently discretize the equation using Fourier collocation. In this way, we obtain a time-independent dissipative matrix system together with an initial state and a recovery map, so that we can treat the time dependence separately from the choice of autonomous solver.

To complete the lower path in Fig.~\ref{fig:autonomization-relations}, we adopt the Schr\"odingerization approach developed in~\cite{JLY22SchrShort,JLY22SchrLong,JLMPY2025schr} for the time-independent non-unitary dynamics. We introduce a second auxiliary variable $p$, use periodic initialization and Fourier projection, and recover the physical state by weighted projection in $p$ followed by selection of the clock node $s=T$. Under the finite-order regularity assumptions stated below, we obtain a query bound with a single logarithmic precision factor. To establish this bound, we control the errors and normalization factors for both auxiliary variables and account for their joint recovery probability.

Our main contributions are as follows.

\begin{itemize}
\item \textbf{Improved autonomization through the choice of clock profile.}
We design the clock profile by combining two ingredients with complementary roles: a compactly supported cosine window function enforces the endpoint regularity needed for Fourier collocation, while a normalized Dirichlet kernel concentrates the sampled mass near the recovery characteristic. By choosing the kernel bandwidth together with the grid, we preserve exact recovery at the continuous level and keep the discrete profile norm bounded as the target accuracy increases. This gives exponential decay of the clock recovery error without an additional normalization penalty in the recovery probability; see Theorem~\ref{thm:target-node}.

\item \textbf{Improved initialization and simpler analysis for Schr\"odingerization.}
Building on the error-function initialization of Ref.~\cite{JLMPY2025schr}, we choose its exact periodization as the auxiliary initial profile. The resulting Fourier coefficients have a closed form, allowing us to estimate the projection error directly from their Gaussian decay. Together with weighted $L^2$ recovery, this simplifies the earlier analysis based on high-order derivative and boundary-truncation estimates, and gives an exact finite-sum implementation of the recovery integral. We retain bounded initialization and recovery norms and establish the autonomous solver in Theorem~\ref{thm:autonomous-Sch-complexity}.

\item \textbf{Single-logarithmic precision dependence for time-dependent systems.}
Combining the two improved constructions, we prepare the normalized final state to accuracy $\varepsilon$ with success probability $\Theta(1)$ using
\[
\mathcal O\Bigl(g\alpha_AT\log\frac{g\alpha_AT}{\varepsilon}\Bigr)
\]
matrix queries and $\mathcal O(g)$ queries to each state-preparation oracle, where $g=\|\bb{x}_0\|/\|\bb{x}(T)\|$; see Theorem~\ref{thm:Sch-auto-complexity}. Thus, we obtain a first power of $\log(1/\varepsilon)$ for time-dependent coefficients, matching the logarithmic precision scaling achieved by the autonomous LCHS and Schr\"odingerization algorithms in Refs.~\cite{low2025optimal,JLMPY2025schr}. The additive contribution of the two frequency cutoffs to the Hamiltonian normalization avoids multiplying their logarithmic costs, while joint recovery gives success probability $\Theta(g^{-2})$ before amplitude amplification.
\end{itemize}

We assume sampled matrix access, exact preparation of the initialization and recovery states, and derivative bounds through order $r-1$, where $r=\Theta(\log(g\alpha_AT/\varepsilon))$. Thus, for each prescribed accuracy, we require only finitely many derivatives, with no derivative queries. For comparison, the non-unitary LCHS construction in Ref.~\cite{ZhaoLiAn2026slow} also achieves a multiplicative bound with a single logarithmic precision factor for slowly varying analytic coefficients $A(t)=B(t/T)$, where the analytic derivative bounds for the fixed path $B$ are independent of $T$, under derivative-access or structured-control assumptions. Our analysis does not require the time variation to decrease with $T$. We compare the bounds under their respective assumptions and access models in Table~\ref{tab:query-comparison}.

We organize the paper as follows. In Section~\ref{sec:problem}, we state the assumptions and main result. In Section~\ref{sec:auto}, we construct the autonomous system and establish the Fourier collocation error and clock normalization estimates. We develop Schr\"odingerization for autonomous non-unitary dynamics in Section~\ref{sec:Sch-auto-complexity} and combine it with autonomization to derive the query bound in Section~\ref{sec:error-complexity}. In Section~\ref{sec:inhomogeneous}, we discuss time-dependent and constant source terms, with separate matrix and source queries. We conclude with a numerical experiment and discussion in Sections~\ref{sec:numerical} and~\ref{sec:conclusion}.

\paragraph{Norms and function spaces.}
Throughout the paper, $\|\cdot\|$ denotes the Euclidean norm for vectors and the induced operator norm for matrices. The inner product $\langle\bb{u},\bb{v}\rangle=\bb{u}^\dagger\bb{v}$ is conjugate-linear in the first argument.

$\|\cdot\|_{L^2(I)}$ and $\|\cdot\|_{L^\infty(I)}$ denote function norms on an interval $I$, using absolute values for scalars, Euclidean norms for vectors, and induced operator norms for matrices. We abbreviate them as $\|\cdot\|_2$ and $\|\cdot\|_\infty$ when the domain is clear. The $L^2$ inner product is $\langle\bb{u},\bb{v}\rangle_{L^2(I)}=\int_I\bb{u}(s)^\dagger\bb{v}(s)\,\d s$, with the scalar version defined analogously. For a scalar function $f$ and a vector-valued function $\bb{v}$, we also use the pairing $\langle f,\bb{v}\rangle_{L^2(I)}=\int_I\overline{f(s)}\bb{v}(s)\,\d s$.
For an integer $r\ge1$, $H^r(I)$ is the Sobolev space of functions with square-integrable weak derivatives through order $r$. For $I=(a,b)$, we write
\[
H_{\mathrm{per}}^r(I)=\bigl\{f\in H^r(I):f^{(j)}(a)=f^{(j)}(b),\ 0\le j\le r-1\bigr\},
\]
where endpoint values are understood as traces. The vector-valued space $H_{\mathrm{per}}^r(I;\mathbb C^N)$ is defined componentwise; we write $\|\bb{f}\|_{H^r(I)}$ for its norm when the value space and periodicity are clear.

\paragraph{Asymptotic notation.}
For nonnegative quantities $f$ and $h$, we write $f=\mathcal O(h)$ if $f\le Ch$, $f=\Omega(h)$ if $f\ge ch$, and $f=\Theta(h)$ if both bounds hold, for positive constants $c,C$ independent of the varying parameters in the stated regime. The notation $\widetilde{\mathcal O}$ additionally suppresses polylogarithmic factors in the relevant parameters; in literature comparisons, these factors are specific to the cited bounds. In particular, success probability $\Theta(1)$ means a probability bounded below by a positive constant independent of the problem parameters. Unless stated otherwise, constants in our estimates may depend on $R_s$, $\rho$, and the fixed profile parameters introduced below, but not on $T$, $\alpha_A$, $g$, $\varepsilon$, or the discretization parameters. We use $C>0$ for such a constant, which may change between estimates.

\section{Assumptions and main result}
\label{sec:problem}

For \eqref{ODElinear} with $A(t)\in\mathbb C^{N\times N}$, $N=2^n$, and $\bb x_0\ne\bb 0$, we seek to prepare $\bb x(T)/\|\bb x(T)\|$ to accuracy $\varepsilon$, with a success flag. We assume $\bb x(T)\ne\bb 0$ and define
\[
g:=\frac{\|\bb{x}_0\|}{\|\bb{x}(T)\|}.
\]
This norm ratio enters the recovery cost.

\subsection{Coefficient assumptions}
\label{subsec:problem-assumptions}

To construct the clock formulation in Section~\ref{sec:auto}, we need coefficient values beyond $[0,T]$. We fix $R_s>0$, independently of $T$, $\alpha_A$, $g$, and $\varepsilon$, and define
\begin{equation}
\label{sinterval}
\Omega_s'=[-R_s,T+R_s], \qquad \Omega_s=[-R_s,R_s+2T), \qquad \ell_s=2(R_s+T).
\end{equation}
The interval $\Omega_s'$ contains the characteristics carrying nonzero data, while $\Omega_s$ is the periodic computational interval. We require an extension of $A$ to $\Omega_s'$, denoted by the same symbol, with the following properties.

\paragraph{Dissipativity.}
The extended coefficient satisfies
\begin{equation}
\label{eq:A-dissipative}
\frac{A(s)+A(s)^\dagger}{2}\preceq O, \qquad s\in\Omega_s'.
\end{equation}
In particular, the physical solution satisfies $\|\bb{x}(t)\|\le\|\bb{x}_0\|$ for $0\le t\le T$, so $g\ge1$. Dissipativity also controls the propagation of errors in the clock discretization.

\paragraph{Finite-order regularity.}
For the approximation order $r\ge2$ chosen in Theorem~\ref{thm:Sch-auto-complexity}, assume $A\in C^{r-1}(\Omega_s';\mathbb C^{N\times N})$ and
\begin{equation}
\label{eq:A-finite-regularity}
\sup_{s\in\Omega_s'}\|A^{(j)}(s)\|\le M_A\Bigl(\frac{2}{\rho}\Bigr)^j j!, \qquad 0\le j\le r-1.
\end{equation}
Here $M_A>0$ controls the magnitude of the coefficient, while $\rho>0$ controls the growth of its higher derivatives. The scale $\rho$ is fixed independently of $T$, $\alpha_A$, $g$, $\varepsilon$, and $r$, and the same constants $M_A$ and $\rho$ apply through the required order.

The factorial growth in \eqref{eq:A-finite-regularity} has the form of a Gevrey-1 derivative bound, which characterizes analytic regularity when imposed at all orders. Here we require it only through order $r-1$ for the prescribed accuracy, so this finite-order condition does not by itself imply analyticity. Higher accuracy requires the same bounds through the corresponding higher order, with the same constants.

We use derivatives of $A$ through order $r-1$ to control the clock solution in $H^r$ in Lemma~\ref{lem:z-Hr} and obtain the exponential error bound in Theorem~\ref{thm:target-node}. We need these derivatives only for the analysis; the algorithm queries sampled matrix values.

On the rest of the computational interval, we set the coefficient to zero:
\begin{equation}
\label{eq:A-zero-extension}
\widetilde A(s)=\begin{cases}A(s),&s\in\Omega_s',\\ O,&s\in\Omega_s\setminus\Omega_s'.\end{cases}
\end{equation}
We identify the endpoints of $\Omega_s$ periodically. No smoothness is required across the endpoints of $\Omega_s'$; the compact support of the clock profile keeps the nonzero continuous solution within the region where the regular extension applies. The zero extension preserves dissipativity,
\begin{equation}
\label{eq:Atilde-dissipative}
\frac{\widetilde A(s)+\widetilde A(s)^\dagger}{2}\preceq O, \qquad s\in\Omega_s.
\end{equation}

\subsection{Oracle access}
\label{subsec:input-model}

We count matrix and state-preparation queries. A unitary $U_M$ is an $(\alpha,m,\delta)$-block-encoding of $M$~\cite{Gilyen2019QSVD} if
\[
\Bigl\|M-\alpha(\bra{0^m}\otimes I)U_M(\ket{0^m}\otimes I)\Bigr\|\le\delta,
\]
where $m$ is the number of ancilla qubits and $I$ acts on the register representing $M$.

\paragraph{Sampled matrix access.}
For each clock grid used by the algorithm, let $N_s=2^{n_s}$ with $n_s\ge1$, and set $s_l=-R_s+l\ell_s/N_s$ for $0\le l<N_s$. We assume access to a unitary $\mathrm{HAM\text{-}T}_A$ that exactly block-encodes the coefficient samples:
\begin{equation}
\label{eq:HAMT-A}
(\bra{0^{a_A}}\otimes I^{\otimes(n_s+n)})\mathrm{HAM\text{-}T}_A(\ket{0^{a_A}}\otimes I^{\otimes(n_s+n)})
=\sum_{l=0}^{N_s-1}\ket l\bra l\otimes\frac{\widetilde A(s_l)}{\alpha_A},
\end{equation}
with $\alpha_A\ge\max\{1,M_A\}$. Here $a_A$ is the number of block-encoding ancillas, and the clock register stores the sample index $l$. The oracle, its adjoint, and their controlled versions are available; each application counts as one matrix query.

\paragraph{State-preparation access.}
We assume exact preparation of $\ket{\bb{x}_0}$, $\ket{G_s}$, $\ket{\Psi_p}$, and $\ket{\chi_p}$. These are the normalized initial vector, the sampled clock profile, the projected initial profile in the second auxiliary variable, and the corresponding recovery weight. Their definitions are given in Sections~\ref{subsec:clock-normalization} and~\ref{sec:quantum-recovery}. Each preparation oracle and its inverse are available, and each application counts as one state-preparation query. The cost of implementing these oracles in elementary gates is outside the query model considered here.

\subsection{Main result and comparison}
\label{subsec:main-result}

Under the preceding assumptions, we obtain the following query bound by combining autonomization and Schr\"odingerization.

\begin{theorem}
\label{thm:Sch-auto-complexity}
Consider system~\eqref{ODElinear} under the dissipativity condition \eqref{eq:A-dissipative} and the oracle model in Section~\ref{subsec:input-model}. Suppose $\bb{x}(T)\ne\bb{0}$, $\alpha_AT\ge1$, and $0<\varepsilon\le1/4$, and let $g=\|\bb{x}_0\|/\|\bb{x}(T)\|$. Choose a sufficiently large integer $r\ge2$ satisfying
\[
r=\Theta\!\Bigl(\log\frac{g\alpha_AT}{\varepsilon}\Bigr),
\]
where the implied constant depends only on the fixed structural parameters. Assume \eqref{eq:A-finite-regularity} holds through order $r-1$ for this choice. Then there exists a quantum algorithm that, with success probability $\Theta(1)$ and a flag indicating success, prepares a state within $\ell^2$-distance $\varepsilon$ of $\bb{x}(T)/\|\bb{x}(T)\|$. The algorithm uses
\begin{itemize}
\item $\mathcal O\!\bigl(g\alpha_AT\log(g\alpha_AT/\varepsilon)\bigr)$ queries to $\mathrm{HAM\text{-}T}_A$, including controlled and inverse queries;
\item $\mathcal O(g)$ queries to each state-preparation oracle.
\end{itemize}
\end{theorem}

We prove the theorem in Section~\ref{sec:error-complexity}.

\begin{remark}[Query bound and regularity]
For fixed $g$ and $\alpha_AT$, the matrix-query bound is $\mathcal O(\log(1/\varepsilon))$. By amplifying the joint recovery probability $\Theta(g^{-2})$, we incur the factor $g$ without an additional normalization penalty growing with accuracy or grid size. To achieve higher accuracy, we require the derivative bounds through the corresponding higher order.
\end{remark}

\begin{remark}[Precision dependence and optimality]
\label{rem:precision-optimality}
Our bound matches the single-logarithmic precision scaling of the autonomous algorithms in Refs.~\cite{low2025optimal,JLMPY2025schr}. A matching lower bound for the homogeneous state-preparation problem under our assumptions would be needed to establish strict optimality. In particular, the reduction in Proposition~9 of Ref.~\cite{An2022forwarding} concerns inhomogeneous ODEs and gives a conditional lower bound involving both time and precision. It does not directly establish an $\Omega(\log(1/\varepsilon))$ lower bound for the present problem. Likewise, optimality of a propagator block-encoding does not by itself establish optimality of preparing a particular normalized solution state; see Section~1.1.7 of Ref.~\cite{low2025optimal}.
\end{remark}

In Table~\ref{tab:query-comparison}, we compare query bounds for homogeneous time-dependent linear systems under the respective assumptions and access models.

\begin{table}[t]
  \centering
\small \renewcommand{\arraystretch}{1.2} \setlength{\tabcolsep}{5pt}
  \caption{
Query bounds for homogeneous time-dependent linear systems. Here $g=\|\bb{x}_0\|/\|\bb{x}(T)\|$, $\alpha_A$ is the matrix normalization, $\kappa_V$ bounds the eigenvector condition number in the spectral method, and $\beta\in(0,1)$ is fixed. Initial-state queries count calls to the preparation oracle for $\ket{\bb{x}_0}$ or its inverse. The notation $\widetilde{\mathcal O}$ suppresses method-specific polylogarithmic factors.
  }
\label{tab:query-comparison}

  \begin{tabular}{@{}lll@{}}
    \toprule
Method & Matrix queries & Initial-state queries \\
    \midrule

Spectral method~\cite{Childs-Liu-2020} & $\widetilde{\mathcal O}\!\Bigl( g\kappa_V\alpha_A T\, \text{polylog}(1/\varepsilon) \Bigr)$ & $\widetilde{\mathcal O}\!\Bigl( g\kappa_V\alpha_A T\, \text{polylog}(1/\varepsilon) \Bigr)$ \\[1mm]

Truncated Dyson~\cite{BerryCosta2024timeODE} & $\widetilde{\mathcal O}\!\Bigl( g\alpha_A T\log^2(1/\varepsilon) \Bigr)$ & $\mathcal O\!\Bigl( g\alpha_A T\log(1/\varepsilon) \Bigr)$ \\[1mm]

Time marching~\cite{Lin2023timeODE} & $\widetilde{\mathcal O}\!\Bigl( g\alpha_A^2T^2\log(1/\varepsilon) \Bigr)$ & $\mathcal O(g)$ \\[1mm]

LCHS~\cite{ALL2023LCH} & $\widetilde{\mathcal O}\!\Bigl( g^2\alpha_A T/\varepsilon \Bigr)$ & $\mathcal O(g)$ \\[1mm]

Improved LCHS~\cite{ACL2023LCH2} & $\widetilde{\mathcal O}\!\Bigl( g\alpha_A T \log^{1+1/\beta}(1/\varepsilon) \Bigr)$ & $\mathcal O(g)$ \\[1mm]

Slow analytic LCHS~\cite{ZhaoLiAn2026slow} & $\widetilde{\mathcal O}\!\Bigl( g\alpha_AT\log(1/\varepsilon) \Bigr)$ & $\widetilde{\mathcal O}(g)$ \\

    \midrule

\textbf{This work} & { $\widetilde{\mathcal O}\!\Bigl( g\alpha_AT\log\frac{1}{\varepsilon} \Bigr)$} & $\mathcal O(g)$ \\

    \bottomrule
  \end{tabular}
\end{table}

Improved LCHS~\cite{ACL2023LCH2} gives a precision factor $\log^{1+1/\beta}(1/\varepsilon)$ for fixed $\beta\in(0,1)$. The listed time-marching bound~\cite{Lin2023timeODE} has a single logarithmic precision factor but quadratic dependence on $\alpha_AT$. We obtain a single logarithmic precision factor with near-linear dependence on $\alpha_AT$, while using $\mathcal O(g)$ initial-state queries independently of $\alpha_AT$ and $\varepsilon$, as in the listed time-marching, LCHS, and improved LCHS methods.

Ref.~\cite{ZhaoLiAn2026slow} studies Hamiltonians of the form $\widetilde H(t)=H(t/T)$, where $H$ is a fixed analytic matrix-valued function on $[0,1]$ with derivative bounds independent of $T$. The term \emph{slow analytic} means that the same analytic path is traversed over a longer physical time as $T$ increases. In particular,
\[
\frac{\d}{\d t}\widetilde H(t)=\frac1T H'(t/T),
\]
so the variation rate decreases as $T^{-1}$. Keeping the path and its regularity bounds independent of $T$ is essential to this assumption.

For these Hamiltonians, the method achieves nearly additive query complexity $\widetilde{\mathcal O}(\alpha T+\log(1/\varepsilon))$, where $\alpha$ is the Hamiltonian normalization, and an additional gate count polynomial in $\log(1/\varepsilon)$. Its general implementation requires access to the first derivative and higher endpoint derivatives; a structured control decomposition provides an alternative access model. The nearly additive bound concerns unitary evolution.

The non-unitary extension retains the slow-variation and analyticity assumptions. In the sign convention of \eqref{ODElinear}, Corollary~7.1 of~\cite{ZhaoLiAn2026slow} applies to coefficients $A(t)=B(t/T)$, where $B$ is a fixed analytic matrix-valued function on $[0,1]$ with $\|B(s)\|\le\alpha$ and $(B(s)+B(s)^\dagger)/2\preceq O$. Its analytic derivative bounds are independent of $T$, so the physical-time derivatives acquire factors $T^{-j}$. Under the corresponding derivative-access or structured-control oracle assumptions, together with initial-state preparation access, combining Hamiltonian simulation with LCHS gives the multiplicative query bound $\widetilde{\mathcal O}(g\alpha T\log(1/\varepsilon))$. This result for slow analytic dissipative linear ODEs is the one included in Table~\ref{tab:query-comparison}.

For non-unitary linear systems, we obtain the same multiplicative dependence on the evolution scale and a single logarithmic precision factor, under different assumptions. We use sampled coefficient values and allow time variation that need not decrease with $T$, while requiring dissipativity, finite-order derivative bounds, a dissipative extension to the clock interval, and exact auxiliary state-preparation access. We focus on oracle-query complexity and leave the possibility of an additive bound for Hamiltonian dynamics to future work.

\section{Autonomization method for time-dependent linear differential equations}
\label{sec:auto}

We remove the time dependence in \eqref{ODElinear} by introducing a clock coordinate $s$~\cite{CJL23TimeSchr}, from which we recover the original solution along $s=t$. To control both the Fourier collocation error and the recovery probability, we choose a clock profile that combines compact support and endpoint regularity with concentration near the recovery characteristic.

\subsection{Autonomous formulation and recovery}
\label{subsec:auto}

For $a\le b$ in $\Omega_s'$, let $U_A(b,a)=\exp_{\mathcal{T}}(\int_a^bA(\xi)\,\d\xi)$ be the time-ordered exponential. Its endpoint derivatives satisfy
\[
\partial_bU_A(b,a)=A(b)U_A(b,a), \qquad \partial_aU_A(b,a)=-U_A(b,a)A(a).
\]

Following Ref.~\cite{CJL23TimeSchr}, we replace $A(t)$ by multiplication by $\widetilde A(s)$ and add transport at unit speed to obtain the time-independent generator $-\partial_s+\widetilde A(s)$. We choose initial data $\bb z(0,s)=G(s)\bb x_0$ with $G(0)=1$, so that we recover the original solution along the characteristic starting at $s=0$.
\begin{theorem}
\label{thm:autonomizationTheorem}
Let $G\in C^1(\mathbb R)$ be supported in $[-R_s,R_s]$ and satisfy $G(0)=1$, where $R_s>0$ is a constant. For the non-autonomous system \eqref{ODElinear}, introduce the following initial-value problem of an autonomous PDE
\begin{equation}
\label{eq:autonomous_pde}
\begin{cases} \dfrac{\partial}{\partial t}\bb{z}(t,s) + \dfrac{\partial}{\partial s}\bb{z}(t,s) = \widetilde A(s)\bb{z}(t,s),\\[1mm] \bb{z}(0,s) = G(s)\bb{x}_0. \end{cases}
\end{equation}
Then, for $0\le t\le T$,
\begin{equation}
\label{eq:characteristic-formula}
\bb{z}(t,s) = \begin{cases} G(s-t)U_A(s,s-t)\bb{x}_0, & s\in \Omega_t:=[t-R_s,t+R_s],\\ \bb{0}, & s\notin \Omega_t. \end{cases}
\end{equation}
In particular, $\text{supp}\{\bb{z}(t,\cdot)\} \subset \Omega_t$, where $\Omega_t$ is the moving support interval, and
\begin{equation}
\label{eq:characteristic-target}
\bb{z}(t,s = t) = \bb{x}(t)
\end{equation}
 for every $t \in [0,T]$.
\end{theorem}

\begin{proof}
Let $\bb{v}(t;s')=\bb{z}(t,s'+t)$, where $s'\in\mathbb R$ labels the starting point of a characteristic. Using the chain rule and \eqref{eq:autonomous_pde}, we obtain
\[
\frac{\d}{\d t}\bb{v}(t;s') = \widetilde A(s'+t)\bb{v}(t;s'), \qquad \bb{v}(0;s')=G(s')\bb{x}_0.
\]
If $G(s')=0$, then $\bb{v}(t;s')=\bb{0}$. Otherwise, $s'\in[-R_s,R_s]$, so $s'+t\in \Omega_s'$ for $0\le t\le T$. By definition, we have $\widetilde A(s'+t)=A(s'+t)$. The solution can be expressed using the time-ordered exponential as
\begin{align*}
\bb{v}(t; s') & = \exp_{\mathcal{T}}\Big( \int_0^t A(s'+\tau) \d \tau \Big)G(s')\bb{x}_0 \\ & = \exp_{\mathcal{T}}\Big( \int_{s'}^{s'+t} A(\tau) \d \tau \Big)G(s')\bb{x}_0 = U_A(s'+t,s')G(s')\bb{x}_0.
\end{align*}
We recover \eqref{eq:characteristic-formula} by taking $s'=s-t$. With $s'=0$ and $G(0)=1$, we obtain
\[
\bb{z}(t,s = t) = U_A(t,0)\bb{x}_0 = \bb{x}(t),
\]
which proves \eqref{eq:characteristic-target}.
\end{proof}

\begin{remark}
By \eqref{eq:characteristic-formula}, $\text{supp}\{\bb{z}(t,\cdot)\} \subset \Omega_t\subset \Omega_s'\subset \Omega_s$ for fixed $t \in [0,T]$. Thus each characteristic along which the solution is nonzero remains in $\Omega_s'$, where $\widetilde A=A$ by \eqref{eq:A-zero-extension}. The possible lack of smoothness of $\widetilde A$ at the endpoints of $\Omega_s'$ therefore does not affect \eqref{eq:characteristic-formula}.
\end{remark}

\begin{remark}
\label{rem:inhomogeneous-scope}
The characteristic construction in Theorem~\ref{thm:autonomizationTheorem} does not require dissipativity. That assumption is used below to control the propagation of discretization errors and to bound the clock recovery probability.
\end{remark}

\subsection{A compactly supported clock profile}
\label{subsec:clock-profile}

We choose the clock profile to satisfy $G(0)=1$ for exact recovery and impose compact support with vanishing endpoint derivatives for periodic Fourier approximation. We must also keep the sampled norm $g_G^2=\sum_l|G(s_l)|^2$ bounded for efficient postselection; compact support alone is insufficient, because the sum grows under grid refinement for a fixed nonzero profile.

We construct $G_r=W_rD_{M_s}$ by using a window function $W_r$ to impose compact support and endpoint regularity and a normalized Dirichlet kernel $D_{M_s}$ to concentrate the mass near the origin. Choosing the bandwidth together with the grid allows us to retain $G_r(0)=1$ and obtain $g_G=\Theta(1)$, as we show in Section~\ref{subsec:clock-normalization}.

For an integer $r\ge2$, we define the window function
\begin{equation}
\label{eq:Wr}
W_r(s)= \begin{cases} \cos^r\!\Bigl(\dfrac{\pi s}{2R_s}\Bigr), & |s|<R_s,\\ 0, & |s|\ge R_s. \end{cases}
\end{equation}
The window function satisfies $W_r(0)=1$, and its derivatives through order $r-1$ vanish at $s=\pm R_s$, allowing zero extension with the required Sobolev regularity.

For the grid size $N_s$ introduced in Section~\ref{subsec:input-model}, write $\mu_l^s=2\pi(l-N_s/2)/\ell_s$ for $l\in\mathbb Z$, where $\ell_s$ is defined in \eqref{sinterval}. The zero frequency corresponds to $l=N_s/2$. For an integer $M_s\ge1$, let $\mu_{D,\max}:=2\pi M_s/\ell_s$ denote the maximum frequency of the Dirichlet kernel, and define
\begin{equation}
\label{eq:DM}
D_{M_s}(s) =\frac{1}{2M_s+1} \sum_{|\mu_l^s|\le\mu_{D,\max}}\e^{\i\mu_l^s s} =\frac{1}{2M_s+1} \Bigl[1+2\sum_{j=1}^{M_s} \cos\Bigl(\frac{2\pi js}{\ell_s}\Bigr)\Bigr].
\end{equation}
With $D_{M_s}(0)=1$, we define the clock profile by
\begin{equation}
\label{eq:G}
G(s)=G_r(s):=W_r(s)D_{M_s}(s).
\end{equation}
By \eqref{eq:Wr} and \eqref{eq:G}, $G_r(0)=1$, $G_r$ vanishes outside $[-R_s,R_s]$, and its derivatives through order $r-1$ vanish at both endpoints.

We illustrate the two factors and their product in Fig.~\ref{fig:clock-profile}. By multiplying by $W_r$, we retain the central peak of $D_{M_s}$ and suppress oscillations near the support endpoints. We use the derivative and $L^2$ estimates below to control the collocation error and the discrete normalization $g_G$, respectively.

\begin{figure}[htbp]
    \centering
    \includegraphics[width=0.9\textwidth]
    {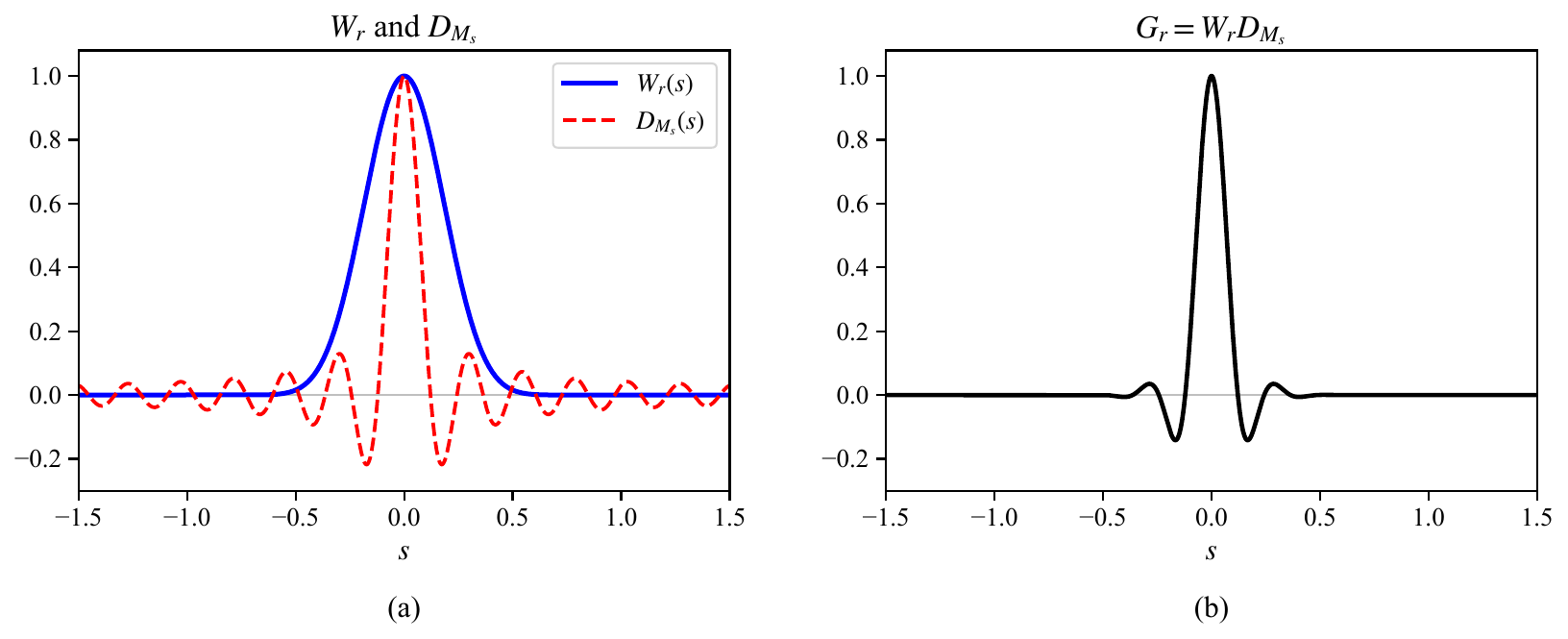}
    \caption{Construction of the clock profile:
    (a) $W_r$ and $D_{M_s}$, (b) $G_r=W_rD_{M_s}$.}
\label{fig:clock-profile}
\end{figure}

The product $G_r$ is not a finite Fourier sum, so we estimate its Fourier error through derivative bounds. On $\Omega_s$, we use the orthonormal basis
\[
\varphi_l(s)=\ell_s^{-1/2}\e^{\i\mu_l^s(s+R_s)}, \qquad l\in\mathbb Z.
\]
For $f=\sum_{l\in\mathbb Z}\widehat f_l\varphi_l$, with $\widehat f_l=\langle\varphi_l,f\rangle_{L^2(\Omega_s)}$, the Sobolev norm has the Fourier representation
\begin{equation}
\label{eq:s-Sobolev-convention}
\|f\|_{H^m(\Omega_s)}^2 =\sum_{l\in\mathbb Z} (1+|\mu_l^s|^2)^m|\widehat f_l|^2.
\end{equation}
We establish the endpoint conditions, derivative bounds, and $L^2$ estimates for $G_r$ in the following lemma.

\begin{lemma}
\label{lem:clock-profile}
Let $r\geq 2$, $M_s\geq 1$, and let $G_r=W_rD_{M_s}$ be defined by \eqref{eq:G}. Then the following statements hold.
\begin{enumerate}[(1)]

    \item For $0\leq m\leq r-1$,
    $G_r^{(m)}(\pm R_s)=0$.

    \item One has $G_r\in H_{\mathrm{per}}^r(\Omega_s)$, and for
    $0\leq m\leq r$,
\begin{equation}
\label{eq:Gr-derivative}
\|G_r^{(m)}\|_{L^2(\Omega_s)} \leq \sqrt{\pi}\, \mu_{D,\max}^{-1/2} \Bigl[ \frac{\pi}{2} \Bigl( \frac{r}{R_s} +\mu_{D,\max} \Bigr) \Bigr]^m.
\end{equation}

\item If
$\mu_{D,\max}R_s\geq c_0\sqrt{r}$ for a sufficiently large fixed constant $c_0$, then
\begin{equation}
\label{eq:Gr-mass}
c\mu_{D,\max}^{-1} \leq \|G_r\|_{L^2(\Omega_s)}^2 \leq C\mu_{D,\max}^{-1},
\end{equation}
    where $c,C>0$ are independent of $r$ and $M_s$.

\end{enumerate}
\end{lemma}

\begin{proof}
(1) We first consider the endpoint regularity of $W_r$. Near $s=R_s$,
\[
\cos\Bigl(\frac{\pi s}{2R_s}\Bigr) = (R_s-s)\varphi_+(s), \qquad \varphi_+(s) = \frac{\sin\!\bigl(\pi(R_s-s)/(2R_s)\bigr)}{R_s-s}.
\]
The function $\varphi_+$ extends smoothly to $s=R_s$ with $\varphi_+(R_s)=\pi/(2R_s)$. Hence $W_r(s)=(R_s-s)^r\varphi_+(s)^r$ for $s<R_s$, which gives $W_r^{(j)}(R_s^-)=0$ for $0\leq j\leq r-1$. Repeating the argument at $s=-R_s$, we obtain
\[
W_r^{(j)}(\pm R_s)=0, \qquad 0\leq j\leq r-1.
\]
Since $D_{M_s}$ is smooth, we apply the Leibniz rule to obtain
\[
G_r^{(m)}(s) = \sum_{j=0}^m \binom{m}{j} W_r^{(j)}(s) D_{M_s}^{(m-j)}(s), \qquad 0\leq m\leq r-1.
\]
Therefore $G_r^{(m)}(\pm R_s)=0$ for $0\leq m\leq r-1$.

(2) We next estimate the derivatives of $G_r$. The endpoint conditions established in (1), together with repeated integration by parts against test functions, show that the zero extension of $W_r$ satisfies
\begin{equation}
\label{eq:Wr-traces}
W_r\in C_c^{r-1}(\mathbb R)\cap W^{r,\infty}(\mathbb R), \qquad W_r^{(j)}(\pm R_s)=0, \quad 0\leq j\leq r-1.
\end{equation}
On $(-R_s,R_s)$, we expand $\cos^r(\pi s/(2R_s))$ by the binomial formula and obtain, for $0\leq j\leq r$,
\begin{equation}
\label{eq:Wr-derivative}
|W_r^{(j)}(s)| \leq 2^{-r}\sum_{q=0}^r \binom{r}{q} \Bigl( \frac{\pi|r-2q|}{2R_s} \Bigr)^j \leq \Bigl(\frac{\pi r}{2R_s}\Bigr)^j.
\end{equation}

By \eqref{eq:DM} and the orthogonality of the Fourier modes,
\begin{align}
\|D_{M_s}\|_2^2 &= \frac{\ell_s}{(2M_s+1)^2} \sum_{|\mu_l^s|\leq \mu_{D,\max}}1 = \frac{\ell_s}{2M_s+1},
\label{eq:DM0}
\\ \|D_{M_s}^{(m)}\|_2^2 &= \frac{\ell_s}{(2M_s+1)^2} \sum_{|\mu_l^s|\leq \mu_{D,\max}} |\mu_l^s|^{2m} \leq \frac{\ell_s}{2M_s+1} \mu_{D,\max}^{2m} \leq \pi \mu_{D,\max}^{2m-1}.
\label{eq:DM-derivative}
\end{align}
Applying the weak Leibniz rule together with \eqref{eq:Wr-derivative} and \eqref{eq:DM-derivative}, we obtain
\begin{align*}
\|\partial_s^mG_r\|_2 &\leq \sum_{j=0}^m \binom{m}{j} \|W_r^{(j)}\|_\infty \|D_{M_s}^{(m-j)}\|_2 \\ &\leq \sqrt{\pi}\, \mu_{D,\max}^{-1/2} \sum_{j=0}^m \binom{m}{j} \Bigl(\frac{\pi r}{2R_s}\Bigr)^j \mu_{D,\max}^{m-j} \\ &= \sqrt{\pi}\, \mu_{D,\max}^{-1/2} \Bigl( \frac{\pi r}{2R_s} +\mu_{D,\max} \Bigr)^m.
\end{align*}
Since
\[
\frac{\pi r}{2R_s} +\mu_{D,\max} \leq \frac{\pi}{2} \Bigl( \frac{r}{R_s} +\mu_{D,\max} \Bigr),
\]
this proves \eqref{eq:Gr-derivative}. Moreover, $G_r\in H^r(\Omega_s)$, and by (1) and $\operatorname{supp}G_r\subset[-R_s,R_s]$, its traces through order $r-1$ agree at the identified endpoints of $\Omega_s$. Hence $G_r\in H_{\mathrm{per}}^r(\Omega_s)$.

(3) Since $|W_r|\leq 1$, \eqref{eq:DM0} gives $\|G_r\|_2^2\leq\ell_s/(2M_s+1) \leq\pi/\mu_{D,\max}$, which proves the upper bound in \eqref{eq:Gr-mass}.

For the lower bound, consider $|s|\leq\ell_s/(4\pi M_s)$. Writing $x=\pi s/\ell_s$, the Dirichlet kernel takes the form
\begin{equation}
\label{eq:DM-sine}
D_{M_s}(s) = \frac{\sin((2M_s+1)x)} {(2M_s+1)\sin x} = \frac{\sin((2M_s+1)x)} {(2M_s+1)x} \frac{x}{\sin x},
\end{equation}
where the expression is understood by continuity at $s=0$. On the above interval, $|x|\leq1/(4M_s)$ and $|(2M_s+1)x|\leq3/4$. Since $\sin y/y\geq\sin(3/4)/(3/4)$ for $|y|\leq3/4$ and $x/\sin x\geq1$, we have
\begin{equation}
\label{eq:DM-main-lobe}
D_{M_s}(s)\geq\frac{\sin(3/4)}{3/4}>0, \qquad |s|\leq\frac{\ell_s}{4\pi M_s}.
\end{equation}

The condition $\mu_{D,\max}R_s\geq c_0\sqrt r$ implies
\[
\Bigl|\frac{\pi s}{2R_s}\Bigr| \leq \frac{\ell_s}{8M_sR_s} \leq \frac{\pi}{4c_0\sqrt r} \leq \frac12
\]
for a sufficiently large fixed $c_0$. Using $\log(\cos x)\geq-x^2$ for $|x|\leq1/2$, we obtain
\[
W_r(s) = \exp\!\Bigl[ r\log\cos\!\Bigl(\frac{\pi s}{2R_s}\Bigr) \Bigr] \geq \exp\!\Bigl( -\frac{\pi^2rs^2}{4R_s^2} \Bigr) \geq \exp\!\Bigl( -\frac{\pi^2}{16c_0^2} \Bigr)>0.
\]
Combining this estimate with \eqref{eq:DM-main-lobe}, we obtain
\[
\|G_r\|_2^2 \geq \int_{-\ell_s/(4\pi M_s)}^{\ell_s/(4\pi M_s)} |W_r(s)D_{M_s}(s)|^2\,\d s \geq c\mu_{D,\max}^{-1},
\]
where $c>0$ depends only on the fixed constant $c_0$. This proves the lower bound in \eqref{eq:Gr-mass}.
\end{proof}

\subsection{Fourier collocation and error analysis}
\label{subsec:clock-discretization}

We discretize the clock equation by Fourier collocation to obtain an autonomous matrix system, then use Sobolev estimates and the sampled-profile norm to choose a grid that controls both recovery error and probability.

\subsubsection{Fourier collocation}

We discretize \eqref{eq:autonomous_pde} on $\Omega_s=[-R_s,R_s+2T)$ using the even grid size $N_s$ and spacing $\Delta s=\ell_s/N_s$:
\[
-R_s=s_0<s_1<\cdots<s_{N_s}=R_s+2T, \qquad s_l=-R_s+l\Delta s.
\]
We identify $s_0$ and $s_{N_s}$ to impose periodicity. Since $N_s$ is even, we have
\[
s_{N_s/2} =-R_s+\frac{N_s}{2}\Delta s =T.
\]
Define
\[
    \bb{w}(t)
    = \sum_{l=0}^{N_s-1}\ket{l}\otimes\bb{w}_l(t),
\]
where $\bb{w}_l(t)$ is the approximation of $\bb{z}(t,s_l)$. By applying Fourier collocation in $s$, we obtain the grid system
\begin{equation}
\label{eq:s-semidiscrete-system}
\dfrac{\d}{\d t}\bb{w}(t)=\bar{A}\bb{w}(t), \qquad \bb{w}(0) =  [G(s_0),\cdots, G(s_{N_s-1})]^\top \otimes\bb x_0.
\end{equation}
Here, the evolution matrix is
\begin{equation}
\label{autoA}
    \bar{A} = -\i P_s \otimes I^{\otimes n} + \sum_{l=0}^{N_s-1} \ket{l}\bra{l} \otimes \widetilde{A}(s_l), \qquad N = 2^n,
\end{equation}
which is {\it time-independent}, with
\begin{equation}
\label{Ps}
P_s = F_s D_s F_s^\dagger, \qquad D_s = \diag(\mu_0^{s}, \cdots, \mu_{N_s-1}^{s}), \qquad \mu_l^{s} = \frac{2\pi}{\ell_s}\Bigl(l-\frac{N_s}{2}\Bigr),
\end{equation}
where $0\le l<N_s$ and $F_s$ maps Fourier coefficients to grid values, as specified in \eqref{eq:s-fourier-convention}. Its adjoint $F_s^\dagger$ maps grid values to Fourier coefficients. The largest Fourier mode in absolute value satisfies
\begin{equation}
\label{eq:mu_s_max_log_square_condition}
	\mu_{s,\max} = \max_{0\le l\le N_s-1}|\mu_l^{s}| = \frac{\pi N_s}{\ell_s} = \frac{\pi}{\Delta s}.
\end{equation}
Let $\bb{x}_s(T)=\bb{w}_{N_s/2}(T)$. By \eqref{eq:characteristic-target}, its error relative to $\bb x(T)$ is the collocation error at $s=T$.

Let $\mathcal K_s=\{0,\ldots,N_s-1\}$ index the retained modes. With the centered frequencies in \eqref{Ps},
\begin{equation}
\label{eq:s-fourier-convention}
F_s\ket{l} =\frac{1}{\sqrt{N_s}}\sum_{j=0}^{N_s-1} \e^{\i\mu_l^s(s_j+R_s)}\ket{j}, \qquad F_s^\dagger\ket{j} =\frac{1}{\sqrt{N_s}}\sum_{l=0}^{N_s-1} \e^{-\i\mu_l^s(s_j+R_s)}\ket{l}.
\end{equation}

\subsubsection{Collocation error and solution regularity}

For the error analysis, we use the Fourier projection onto these modes,
\[
\Pi_s f=\sum_{l\in\mathcal K_s}\widehat f_l\varphi_l.
\]
We estimate the omitted modes using the Sobolev norm in \eqref{eq:s-Sobolev-convention}. To relate these continuous estimates to grid errors, for grid vectors with blocks \(\bb{u}_l,\bb{v}_l\in\mathbb C^N\), we define
\[
\langle\bb{u},\bb{v}\rangle_h =\Delta s\sum_{l=0}^{N_s-1}\bb{u}_l^\dagger\bb{v}_l, \qquad \|\bb{v}\|_h=\sqrt{\Delta s}\,\|\bb{v}\|.
\]
In the following lemma, we bound the recovery-node error in terms of $\|\bb{z}(t,\cdot)\|_{H^r(\Omega_s)}$.

\begin{lemma}
\label{lem:clock-error}
Assume $\bb{z}(t,\cdot)\in H_{\mathrm{per}}^r(\Omega_s;\mathbb C^N)$ for $0\le t\le T$ and an integer $r\ge2$. Then
\[
\|\bb{x}_s(T)-\bb{x}(T)\| \le CT(\Delta s)^{-1/2}\mu_{s,\max}^{1-r} \sup_{0\le t\le T}\|\bb{z}(t,\cdot)\|_{H^r(\Omega_s)},
\]
where $C$ is independent of $r$, $T$, and $N_s$.
\end{lemma}
\begin{proof}
 Let $\mathcal{S}_N$ be an operator defined by $\mathcal{S}_N f=\sum_{l=0}^{N_s-1}\ket{l}\otimes f(s_l)$. Then $\bb{w}(t)\approx\mathcal{S}_N\bb{z}(t,\cdot)$. To compare the two vectors, we first justify evaluation of the Fourier derivative series at the grid nodes. Since $\partial_s\varphi_n=\i\mu_n^s\varphi_n$, the Fourier coefficients of the weak derivative satisfy $\widehat{\partial_s\bb{z}}_n(t) =\i\mu_n^s\widehat{\bb{z}}_n(t)$. Thus
\[
\lim_{K\to\infty} \Bigl\|\partial_s\bb{z}(t,\cdot) -\sum_{|\mu_n^s|\le K} \i\mu_n^s\widehat{\bb{z}}_n(t)\varphi_n \Bigr\|_{L^2(\Omega_s)}=0.
\]
For $r\ge2$, we obtain by Cauchy--Schwarz
\[
\sum_{n\ne N_s/2}|\mu_n^s|\,\|\widehat{\bb{z}}_n(t)\| \le \Bigl(\sum_{n\ne N_s/2}|\mu_n^s|^{2-2r}\Bigr)^{1/2} \|\partial_s^r\bb{z}(t,\cdot)\|_2<\infty.
\]
Together with $|\varphi_n(s)|=\ell_s^{-1/2}$, this proves absolute and uniform convergence of the derivative series. Its sum agrees with the weak derivative in $L^2(\Omega_s)$ and therefore represents its continuous representative. In particular,
\[
\partial_s\bb{z}(t,s_l) =\sum_{n\in\mathbb Z} \i\mu_n^s\widehat{\bb{z}}_n(t)\varphi_n(s_l), \qquad 0\le l<N_s.
\]
Evaluating \eqref{eq:autonomous_pde} at the grid nodes, we obtain

\[
\frac{\d}{\d t}\mathcal{S}_N\bb{z} = -\mathcal{S}_N \partial_s\bb{z} + \Big(\sum_{l=0}^{N_s-1} \ket{l}\bra{l}\otimes\widetilde A(s_l)\Big) \mathcal{S}_N\bb{z} = -\mathcal{S}_N \partial_s\bb{z} + (\bar{A} + \i P_s\otimes I^{\otimes n} ) \mathcal{S}_N\bb{z}.
\]
Set $\bb{e}_N(t):=\bb{w}(t)-\mathcal S_N\bb{z}(t,\cdot)$. Subtracting the sampled equation from \eqref{eq:s-semidiscrete-system}, we obtain
\[
\frac{\d}{\d t} \bb{e}_N = \bar{A}\bb{e}_N - (\i (P_s \otimes I^{\otimes n}) \mathcal{S}_N-\mathcal{S}_N\partial_s )\bb{z} = \bar{A}\bb{e}_N-\bb{\tau}_N, \qquad \bb{e}_N(0)=\bb{0},
\]
where the collocation residual is
\[
\bb{\tau}_N(t) :=\bigl(\i(P_s\otimes I^{\otimes n})\mathcal S_N -\mathcal S_N\partial_s\bigr)\bb{z}(t,\cdot).
\]
By variation of constants, we obtain $\bb{e}_N(T)=-\int_0^T \e^{(T-t)\bar A}\bb{\tau}_N(t)\,\d t$. The propagation of this residual is controlled by dissipativity. Indeed, $P_s=P_s^\dagger$ and \eqref{eq:Atilde-dissipative} imply
\[
\operatorname{Re}\langle\bb{v},\bar A\bb{v}\rangle_h =\Delta s\sum_{l=0}^{N_s-1} \Bigl\langle\bb{v}_l, \frac{\widetilde A(s_l)+\widetilde A(s_l)^\dagger}{2} \bb{v}_l\Bigr\rangle \le0.
\]
Consequently, $\bb{y}(t)=\e^{t\bar A}\bb{v}$ satisfies
\[
\frac{\d}{\d t}\|\bb{y}(t)\|_h^2 =2\operatorname{Re} \langle\bb{y}(t),\bar A\bb{y}(t)\rangle_h\le0.
\]
Thus $\|\e^{t\bar A}\bb{v}\|_h\le\|\bb{v}\|_h$ for $t\ge0$, and the error is bounded by
\[
\|\bb{e}_N(T)\|_h \le\int_0^T\|\bb{\tau}_N(t)\|_h\,\d t.
\]
It remains to estimate $\bb{\tau}_N$. By \eqref{Ps} and \eqref{eq:s-fourier-convention}, $F_s^\dagger P_s=D_sF_s^\dagger$. Moreover, $\varphi_{l+mN_s}(s_j)=\varphi_l(s_j)$. By absolute convergence, we can regroup the Fourier series and use discrete orthogonality to obtain
\[
\sqrt{\Delta s}\, \bigl((F_s^\dagger\otimes I^{\otimes n}) \mathcal S_N\bb{z}\bigr)_l =\sum_{m\in\mathbb Z}\widehat{\bb{z}}_{l+mN_s}, \qquad l\in\mathcal K_s.
\]
Applying the same identity to $\partial_s\bb{z}$, we obtain
\[
\begin{aligned} \sqrt{\Delta s}\, \bigl((F_s^\dagger\otimes I^{\otimes n})\bb{\tau}_N\bigr)_l &=\i\mu_l^s\sum_{m\in\mathbb Z}\widehat{\bb{z}}_{l+mN_s} -\sum_{m\in\mathbb Z} \i\mu_{l+mN_s}^s\widehat{\bb{z}}_{l+mN_s}\\ &=\sum_{m\ne0} \i(\mu_l^s-\mu_{l+mN_s}^s)\widehat{\bb{z}}_{l+mN_s}. \end{aligned}
\]
The term with $m=0$ cancels, so only aliased modes contribute to the residual. For $l\in\mathcal K_s$ and $m\ne0$,
\[
|\mu_l^s-\mu_{l+mN_s}^s| =2|m|\mu_{s,\max}, \qquad |\mu_{l+mN_s}^s| \ge(2|m|-1)\mu_{s,\max}.
\]
We therefore obtain by weighted Cauchy--Schwarz
\[
\begin{aligned} \Delta s\, \bigl\|\bigl((F_s^\dagger\otimes I^{\otimes n}) \bb{\tau}_N\bigr)_l\bigr\|^2 &\le \Bigl(\sum_{m\ne0} \frac{|\mu_l^s-\mu_{l+mN_s}^s|^2} {|\mu_{l+mN_s}^s|^{2r}}\Bigr) \sum_{m\ne0}|\mu_{l+mN_s}^s|^{2r} \|\widehat{\bb{z}}_{l+mN_s}\|^2\\ &\le4\mu_{s,\max}^{2-2r} \Bigl(\sum_{m\ne0}\frac{m^2}{(2|m|-1)^{2r}}\Bigr) \sum_{m\ne0}|\mu_{l+mN_s}^s|^{2r} \|\widehat{\bb{z}}_{l+mN_s}\|^2. \end{aligned}
\]
The constant is uniform in $r$, since
\[
\sup_{r\ge2}\sum_{m\ne0} \frac{m^2}{(2|m|-1)^{2r}} \le\sum_{m\ne0}\frac{1}{m^2} =\frac{\pi^2}{3}.
\]
Each omitted index has a unique representation $n=l+mN_s$ with $l\in\mathcal K_s$ and $m\ne0$. Summing over $l$ and using the unitarity of $F_s^\dagger$ and Parseval's identity, we obtain
\[
\begin{aligned} \|\bb{\tau}_N(t)\|_h^2 &=\Delta s\sum_{l\in\mathcal K_s} \bigl\|\bigl((F_s^\dagger\otimes I^{\otimes n}) \bb{\tau}_N(t)\bigr)_l\bigr\|^2\\ &\le C\mu_{s,\max}^{2-2r} \sum_{n\notin\mathcal K_s} |\mu_n^s|^{2r}\|\widehat{\bb{z}}_n(t)\|^2\\ &\le C\mu_{s,\max}^{2-2r} \|\bb{z}(t,\cdot)\|_{H^r(\Omega_s)}^2. \end{aligned}
\]
Combining the residual estimate with contractivity, we obtain
\[
\|\bb{e}_N(T)\|_h \le CT\mu_{s,\max}^{1-r} \sup_{0\le t\le T}\|\bb{z}(t,\cdot)\|_{H^r(\Omega_s)}.
\]
Finally, using \eqref{eq:characteristic-target} and the weighted grid norm, we obtain
\[
\begin{aligned} \|\bb{x}_s(T)-\bb{x}(T)\| &=\|(\bb{e}_N(T))_{N_s/2}\| \le(\Delta s)^{-1/2}\|\bb{e}_N(T)\|_h\\ &\le CT(\Delta s)^{-1/2}\mu_{s,\max}^{1-r} \sup_{0\le t\le T}\|\bb{z}(t,\cdot)\|_{H^r(\Omega_s)}. \end{aligned}
\]
This proves the claim.
\end{proof}

To apply Lemma~\ref{lem:clock-error}, we need a bound on $\|\bb{z}(t,\cdot)\|_{H^r(\Omega_s)}$. We use the characteristic representation \eqref{eq:characteristic-formula} and differentiate $U_A(s,s-t)$ with respect to both endpoints. In doing so, we need derivatives of $A$ only through order $r-1$ to estimate propagator derivatives through order $r$, which explains \eqref{eq:A-finite-regularity}. We also use dissipativity to bound the propagator by one and avoid exponential growth in $T$.

\begin{lemma}
\label{lem:z-Hr}
Assume \eqref{eq:A-finite-regularity} and \eqref{eq:A-dissipative}. Then, for every integer $0\le j\le r$,
\begin{equation}
\label{eq:z-Hr}
\sup_{0\le t\le T}\|\bb{z}(t,\cdot)\|_{H^j(\Omega_s)} \le C\mu_{D,\max}^{-1/2} \Bigl[1+\frac{\pi}{2}\Bigl(\frac{r}{R_s} +\mu_{D,\max}\Bigr)+2M_A+\frac{2j}{\rho}\Bigr]^j \|\bb{x}_0\|.
\end{equation}
Moreover, $\bb{z}(t,\cdot)\in H_{\mathrm{per}}^r(\Omega_s;\mathbb C^N)$ for $0\le t\le T$.

\end{lemma}

\begin{proof}
Using \eqref{eq:characteristic-formula}, we express $\bb{z}$ as the product of the transported clock profile and a propagator. To estimate its derivatives, we set $Q_t(s):=U_A(s,s-t)$ for $s\in\Omega_t:=[t-R_s,t+R_s]$ and $0\le t\le T$. Dissipativity gives $\|Q_t(s)\|\le1$, while differentiation with respect to the two endpoints gives $\partial_sQ_t(s)=A(s)Q_t(s)-Q_t(s)A(s-t)$. Since $s,s-t\in\Omega_s'$ for $s\in\Omega_t$, \eqref{eq:A-finite-regularity} applies to both coefficient matrices. Differentiating $m$ more times yields
\[
\partial_s^{m+1}Q_t(s) =\sum_{\ell=0}^m\binom{m}{\ell} \Bigl[A^{(\ell)}(s)\partial_s^{m-\ell}Q_t(s) -\partial_s^{m-\ell}Q_t(s)A^{(\ell)}(s-t)\Bigr], \qquad 0\le m\le r-1.
\]
To account for the factorials in \eqref{eq:A-finite-regularity}, we define
\[
q_m:=\frac{1}{m!}\sup_{0\le t\le T} \|\partial_s^mQ_t\|_{L^\infty(\Omega_t)}, \qquad 0\le m\le r.
\]
Taking norms in the differentiated identity and using $\binom{m}{\ell}\ell!(m-\ell)!=m!$, we obtain $q_0\le1$ and
\begin{equation}
\label{eq:Qt-normalized-recurrence}
(m+1)q_{m+1} \le2M_A\sum_{\ell=0}^m \Bigl(\frac{2}{\rho}\Bigr)^\ell q_{m-\ell}, \qquad 0\le m\le r-1.
\end{equation}
We obtain an explicit bound by comparison with the sequence defined by $b_0=1$ and $(m+1)b_{m+1}=2M_A\sum_{\ell=0}^m(2/\rho)^\ell b_{m-\ell}$ for $m\ge0$. The coefficients are nonnegative, so \eqref{eq:Qt-normalized-recurrence} gives $q_m\le b_m$ by induction for $0\le m\le r$. The generating function $b(z)=\sum_{m\ge0}b_mz^m$ satisfies
\[
b'(z)=\frac{2M_A}{1-2z/\rho}\,b(z), \qquad b(0)=1, \qquad b(z)=\Bigl(1-\frac{2z}{\rho}\Bigr)^{-M_A\rho}.
\]
From the binomial expansion, we obtain $m!b_m=\prod_{\ell=0}^{m-1}(2M_A+2\ell/\rho)$, and hence
\begin{equation}
\label{eq:propagator-derivative-bound}
\sup_{0\le t\le T} \|\partial_s^mQ_t\|_{L^\infty(\Omega_t)} \le\prod_{\ell=0}^{m-1}\Bigl(2M_A+\frac{2\ell}{\rho}\Bigr) \le\Bigl(2M_A+\frac{2m}{\rho}\Bigr)^m, \qquad 0\le m\le r,
\end{equation}
where the empty product is $1$. We now apply this bound to $\bb{z}(t,s)=Q_t(s)G_r(s-t)\bb{x}_0$ on $\Omega_t$. Set $a_G:=\frac{\pi}{2}(r/R_s+\mu_{D,\max})$. The Leibniz rule, \eqref{eq:propagator-derivative-bound}, and \eqref{eq:Gr-derivative} give, for $0\le q\le r$,
\begin{equation}
\label{eq:z-derivative-bound}
\begin{aligned} \|\partial_s^q\bb{z}(t,\cdot)\|_{L^2(\Omega_t)} &\le C\mu_{D,\max}^{-1/2} \sum_{m=0}^q\binom{q}{m} \Bigl(2M_A+\frac{2m}{\rho}\Bigr)^m a_G^{q-m}\|\bb{x}_0\|\\ &\le C\mu_{D,\max}^{-1/2} \Bigl(a_G+2M_A+\frac{2q}{\rho}\Bigr)^q\|\bb{x}_0\|. \end{aligned}
\end{equation}
The second inequality follows from $2M_A+2m/\rho\le2M_A+2q/\rho$ and the binomial formula. To extend \eqref{eq:z-derivative-bound} to $\Omega_s$, it remains to check the traces at the endpoints of $\Omega_t$. By \eqref{eq:Wr-traces} and $G_r=W_rD_{M_s}$, $G_r^{(\ell)}(\pm R_s)=0$ for $0\le\ell\le r-1$. We then apply the Leibniz rule to obtain
\[
\partial_s^q[Q_t(s)G_r(s-t)]\big|_{s=t\pm R_s}=0, \qquad 0\le q\le r-1.
\]
Consequently, repeated integration by parts in \eqref{eq:characteristic-formula} introduces no boundary terms, and the weak derivatives satisfy
\[
\partial_s^q\bb{z}(t,s) =\mathbf1_{\Omega_t}(s) \partial_s^q[Q_t(s)G_r(s-t)\bb{x}_0], \qquad 0\le q\le r.
\]
Thus $\bb{z}(t,\cdot)\in H^r(\Omega_s;\mathbb C^N)$, and \eqref{eq:z-derivative-bound} also holds with $\Omega_t$ replaced by $\Omega_s$. The support condition $\operatorname{supp}\bb{z}(t,\cdot)\subset\Omega_t$ and the vanishing traces further imply
\[
\partial_s^q\bb{z}(t,-R_s^+) =\partial_s^q\bb{z}\bigl(t,(R_s+2T)^-\bigr)=0, \qquad 0\le q\le r-1.
\]
The traces therefore agree at the identified endpoints of $\Omega_s$, which proves $\bb{z}(t,\cdot)\in H_{\mathrm{per}}^r(\Omega_s;\mathbb C^N)$. Finally, fix $0\le j\le r$ and set $B_j:=a_G+2M_A+2j/\rho$. For $0\le q\le j$, \eqref{eq:z-derivative-bound} bounds $\|\partial_s^q\bb{z}(t,\cdot)\|_{L^2(\Omega_s)}$ by $C\mu_{D,\max}^{-1/2}B_j^q\|\bb{x}_0\|$. Using Parseval's identity and \eqref{eq:s-Sobolev-convention}, we obtain
\[
\begin{aligned} \|\bb{z}(t,\cdot)\|_{H^j(\Omega_s)}^2 &=\sum_{q=0}^j\binom{j}{q} \|\partial_s^q\bb{z}(t,\cdot)\|_{L^2(\Omega_s)}^2\\ &\le C\mu_{D,\max}^{-1} (1+B_j^2)^j\|\bb{x}_0\|^2\\ &\le C\mu_{D,\max}^{-1} (1+B_j)^{2j}\|\bb{x}_0\|^2. \end{aligned}
\]
Taking square roots and the supremum over $0\le t\le T$ proves \eqref{eq:z-Hr}.
\end{proof}

\subsubsection{Normalization of the sampled clock profile}
\label{subsec:clock-normalization}

We prepare the sampled clock profile as
\begin{equation}
\label{eq:clock-profile-state}
\ket{G_s}=\frac{1}{g_G}\sum_{l=0}^{N_s-1}G(s_l)\ket l, \qquad g_G^2:=\sum_{l=0}^{N_s-1}|G(s_l)|^2.
\end{equation}
The factor $g_G$ enters the clock recovery probability analyzed in Section~\ref{subsec:clock-success}. To estimate it, we relate the sampled norm to the continuous mass bound \eqref{eq:Gr-mass} through Fourier interpolation. Define
\begin{equation}
\label{eq:clock-error-factors}
\vartheta_G :=\frac{\frac{\pi}{2}(r/R_s+\mu_{D,\max})}{\mu_{s,\max}}, \qquad \vartheta_s :=\vartheta_G+\frac{1+2\alpha_A+2r/\rho}{\mu_{s,\max}}.
\end{equation}
We use $\vartheta_G$ to estimate profile interpolation error and $\vartheta_s$ to include the propagator derivatives. With the former, we relate $g_G^2$ to $N_s/(2M_s)$ in the next lemma.

\begin{lemma}
\label{lem:discrete-clock-normalization}
Assume the hypotheses of Lemma~\ref{lem:clock-profile} and $\mu_{D,\max}R_s\ge c_0\sqrt r$. Let $\vartheta_G$ be defined by \eqref{eq:clock-error-factors}, and let $c$ be the lower constant in \eqref{eq:Gr-mass}. There exists a fixed constant $C\ge\sqrt c$ such that, if $\vartheta_G^r\le\frac{\sqrt c}{2C}$, then
\begin{equation}
\label{eq:clock-normalization-bound}
g_G^2=\Theta\!\Bigl(\frac{N_s}{2M_s}\Bigr).
\end{equation}
The implied constants are independent of $T$, $\alpha_A$, $g$, $\varepsilon$, and $r$.

\end{lemma}
\begin{proof}
We compare the sampled norm with the continuous mass of $G$ through its Fourier interpolant. Write $G=\sum_{n\in\mathbb Z}\widehat G_n\varphi_n$ in $L^2(\Omega_s)$. Since $G\in H_{\mathrm{per}}^r(\Omega_s)$ and $r\ge2$, Cauchy--Schwarz gives
\[
\sum_{n\ne N_s/2}|\widehat G_n| \le \Bigl(\sum_{n\ne N_s/2}|\mu_n^s|^{-2r}\Bigr)^{1/2} \|G^{(r)}\|_2<\infty.
\]
The zero-frequency coefficient is finite, and $|\varphi_n(s)|=\ell_s^{-1/2}$. Hence the Fourier series converges absolutely and uniformly. Its sum agrees with $G$ in $L^2(\Omega_s)$ and therefore equals its continuous representative. The series can thus be evaluated at the grid nodes and regrouped according to the indices modulo $N_s$. Define the Fourier interpolant by
\[
I_NG :=\sum_{l\in\mathcal K_s} \Bigl(\sum_{m\in\mathbb Z}\widehat G_{l+mN_s}\Bigr)\varphi_l.
\]
Since $\varphi_{l+mN_s}(s_j)=\varphi_l(s_j)$, one has $I_NG(s_j)=G(s_j)$ for $0\le j<N_s$. By discrete Fourier orthogonality, we also have
\[
\sqrt{\Delta s}\,(F_s^\dagger\mathcal S_NG)_l =\sum_{m\in\mathbb Z}\widehat G_{l+mN_s}, \qquad l\in\mathcal K_s.
\]

To estimate $I_NG-G$, we separate the Fourier projection error from the contribution of the aliased coefficients. Every mode outside $\mathcal K_s$ has absolute frequency at least $\mu_{s,\max}$, so Parseval's identity gives
\begin{equation}
\label{eq:G-projection-error}
\|G-\Pi_sG\|_2^2 =\sum_{n\notin\mathcal K_s}|\widehat G_n|^2 \le\mu_{s,\max}^{-2r}\|G^{(r)}\|_2^2.
\end{equation}
For $l\in\mathcal K_s$ and $m\ne0$,
\[
|l+mN_s-N_s/2| \ge\Bigl(|m|-\frac12\Bigr)N_s, \qquad |\mu_{l+mN_s}^s| \ge(2|m|-1)\mu_{s,\max}.
\]
Since $\sum_{m\ne0}(2|m|-1)^{-2r}$ is uniformly bounded for $r\ge2$, we obtain by weighted Cauchy--Schwarz
\[
\Bigl|\sum_{m\ne0}\widehat G_{l+mN_s}\Bigr|^2 \le C\mu_{s,\max}^{-2r} \sum_{m\ne0}|\mu_{l+mN_s}^s|^{2r} |\widehat G_{l+mN_s}|^2.
\]
The difference $I_NG-\Pi_sG$ contains precisely these aliased coefficients. Moreover, each $n\notin\mathcal K_s$ has a unique representation $n=l+mN_s$ with $l\in\mathcal K_s$ and $m\ne0$. Applying Parseval's identity, we obtain
\begin{equation}
\label{eq:G-aliasing-error}
\begin{aligned} \|I_NG-\Pi_sG\|_2^2 &=\sum_{l\in\mathcal K_s} \Bigl|\sum_{m\ne0}\widehat G_{l+mN_s}\Bigr|^2\\ &\le C\mu_{s,\max}^{-2r} \sum_{l\in\mathcal K_s}\sum_{m\ne0} |\mu_{l+mN_s}^s|^{2r}|\widehat G_{l+mN_s}|^2
\\ &=C\mu_{s,\max}^{-2r} \sum_{n\notin\mathcal K_s}|\mu_n^s|^{2r}|\widehat G_n|^2
\le C\mu_{s,\max}^{-2r}\|G^{(r)}\|_2^2.
\end{aligned}
\end{equation}
Combining \eqref{eq:G-projection-error} and \eqref{eq:G-aliasing-error} with the triangle inequality, then using \eqref{eq:Gr-derivative} and \eqref{eq:clock-error-factors}, we obtain
\[
\|I_NG-G\|_2 \le C\mu_{s,\max}^{-r}\|G^{(r)}\|_2 \le C\mu_{D,\max}^{-1/2}\vartheta_G^r.
\]
Choose the constant $C$ in the hypothesis at least as large as both $\sqrt c$ and the constant in this interpolation bound. Then $\vartheta_G^r\le\frac{\sqrt c}{2C}$ implies $\vartheta_G^r\le1/2$ and hence $\vartheta_G<1$. It also makes the interpolation error at most half the lower bound for $\|G\|_2$ supplied by \eqref{eq:Gr-mass}. Applying the triangle and reverse triangle inequalities, we obtain
\[
\frac{\sqrt c}{2} \mu_{D,\max}^{-1/2} \le\|I_NG\|_2 \le C\mu_{D,\max}^{-1/2}.
\]
Finally, since $I_NG$ agrees with $G$ at the grid nodes, we use discrete Parseval to obtain
\[
\Delta s\,g_G^2 =\|I_NG\|_2^2 =\Theta\!\Bigl(\mu_{D,\max}^{-1}\Bigr).
\]
Substituting $\Delta s=\ell_s/N_s$ proves \eqref{eq:clock-normalization-bound}.
\end{proof}

\subsection{Recovery error and success probability}
\label{subsec:clock-success}

We combine Lemmas~\ref{lem:z-Hr}, \ref{lem:clock-error}, and~\ref{lem:discrete-clock-normalization} to choose a clock grid that controls both recovery error and success probability. We first establish exponential error decay with bounded profile normalization, and then use these estimates to analyze the clock projection.

\begin{theorem}
\label{thm:target-node}
For a sufficiently large integer $r\ge2$, assume \eqref{eq:A-finite-regularity} and \eqref{eq:A-dissipative}. Choose $N_s=2^{n_s}$ and $M_s\ge1$ so that $N_s/(2M_s)$ is a sufficiently large fixed power of two and
\[
N_s=\Theta\bigl(\ell_s(\alpha_A+r)\bigr),
\]
with a sufficiently large implied constant. Then $g_G=\Theta(1)$,
\begin{equation}
\label{eq:s-frequency-scales}
\mu_{s,\max}=\Theta(\alpha_A+r), \qquad \mu_{D,\max} =\frac{2M_s}{N_s}\mu_{s,\max} =\Theta(\alpha_A+r),
\end{equation}
and
\begin{equation}
\label{eq:target-node-error}
\frac{\|\bb{x}_s(T)-\bb{x}(T)\|}{\|\bb{x}_0\|} \le C_s\alpha_A T\e^{-c_s r}.
\end{equation}
All constants and the minimum required order depend only on the fixed structural parameters.
\end{theorem}
\begin{proof}
Write the lower bound in the grid scaling as $\mu_{s,\max}\ge b_s(\alpha_A+r)$, where $b_s>0$ is a sufficiently large fixed constant. By \eqref{eq:clock-error-factors} and $\alpha_A\ge1$,
\[
\begin{aligned} \vartheta_s &=\frac{\pi M_s}{N_s} +\frac{1+2\alpha_A+ \bigl(\frac{\pi}{2R_s}+\frac{2}{\rho}\bigr)r} {\mu_{s,\max}}\\ &\le\frac{\pi M_s}{N_s} +\frac{3+\frac{\pi}{2R_s}+\frac{2}{\rho}}{b_s}. \end{aligned}
\]
Choose $N_s/(2M_s)\ge4$ to be a sufficiently large power of two, and then choose $b_s$ sufficiently large that
\[
\frac{\pi M_s}{N_s}\le\frac14, \qquad \frac{3+\frac{\pi}{2R_s}+\frac{2}{\rho}}{b_s} \le\frac14.
\]
It follows that $\vartheta_G\le\vartheta_s\le1/2$. With the same grid choice, we also obtain
\[
\mu_{D,\max}R_s =\frac{2M_s}{N_s}R_s\mu_{s,\max} \ge\frac{2M_s}{N_s}R_s b_s(\alpha_A+r).
\]
Since $N_s/(2M_s)$ is fixed, increasing $b_s$ if necessary ensures $\mu_{D,\max}R_s\ge c_0\sqrt r$ for all $r\ge2$. Thus the mass estimate in Lemma~\ref{lem:clock-profile} applies to these parameters. By the assumed grid scaling, $\mu_{s,\max}=\pi N_s/\ell_s=\Theta(\alpha_A+r)$. Multiplying by $2M_s/N_s$ proves \eqref{eq:s-frequency-scales}.

For sufficiently large $r\ge2$, one has $2^{-r}\le\frac{\sqrt c}{2C}$, where $C$ and $c$ are the fixed constants in Lemma~\ref{lem:discrete-clock-normalization}. Since $\vartheta_G\le1/2$, that lemma applies and gives
\[
g_G^2=\Theta\!\Bigl(\frac{N_s}{2M_s}\Bigr)=\Theta(1).
\]

To bound the recovery-node error, we combine Lemmas~\ref{lem:clock-error} and \ref{lem:z-Hr} with $j=r$. Using $M_A\le\alpha_A$, $\Delta s=\ell_s/N_s$, and \eqref{eq:clock-error-factors}, we obtain
\begin{equation}
\label{eq:target-node-parameter-bound}
\begin{aligned} \frac{\|\bb{x}_s(T)-\bb{x}(T)\|}{\|\bb{x}_0\|} &\le CT\mu_{s,\max} \sqrt{\frac{N_s}{2M_s}}\,\vartheta_s^r\\ &\le C\alpha_A T(r+1)2^{-r}\\ &\le C_s\alpha_A T\e^{-c_s r}. \end{aligned}
\end{equation}
The second inequality uses the fixed ratio, \eqref{eq:s-frequency-scales}, and $\alpha_A+r\le\alpha_A(r+1)$. For any fixed $0<c_s<\log2$, the factor $(r+1)2^{-r}$ is bounded by a constant multiple of $\e^{-c_s r}$, which gives the last inequality. This proves \eqref{eq:target-node-error}.
\end{proof}

With the recovery error and profile normalization controlled, we now estimate the probability of selecting the clock node $s=T$ from $\ket{\bb w(T)}:=\bb w(T)/\|\bb w(T)\|$.

Since $s_{N_s/2}=T$, selecting the clock outcome $N_s/2$ has probability
\begin{equation}
\label{eq:clock-success-exact}
\mathbb P_s:=\bigl\|(\bra{N_s/2}_s\otimes I^{\otimes n})\ket{\bb{w}(T)}\bigr\|^2
=\frac{\|\bb{x}_s(T)\|^2}{\|\bb{w}(T)\|^2}
=\frac{\|\bb{w}_{N_s/2}(T)\|^2}{\sum_{l=0}^{N_s-1}\|\bb{w}_l(T)\|^2}.
\end{equation}
Conditioned on this outcome, the physical register is in $\bb{x}_s(T)/\|\bb{x}_s(T)\|$. We retain all $N$ physical components and select only their clock index.

\begin{proposition}
\label{prop:clock-success}
Let $0<\eta<1$, $\bb x(T)\ne\bb 0$, and $g=\|\bb x_0\|/\|\bb x(T)\|$. Choose a sufficiently large integer $r\ge2$ satisfying
\begin{equation}
\label{eq:clock-success-order}
r\ge\frac{1}{c_s}\log\frac{C_sg\alpha_AT}{\eta},
\end{equation}
where $C_s$ and $c_s$ are the constants in Theorem~\ref{thm:target-node}, and the fixed minimum order is understood as in that theorem. Assume the coefficient conditions of that theorem hold through the required order, and choose the clock profile and grid as specified there. Then $g_G=\Theta(1)$ and
\[
\|\bb{x}_s(T)-\bb{x}(T)\|\le\eta\|\bb{x}(T)\|.
\]
The clock recovery probability in \eqref{eq:clock-success-exact} satisfies
\begin{equation}
\label{eq:clock-success-lower}
\mathbb P_s\ge\frac{(1-\eta)^2}{g_G^2g^2}.
\end{equation}
Moreover, retaining the norm decay of the autonomous system gives the two-sided estimate
\begin{equation}
\label{eq:clock-success-weighted}
\frac{(1-\eta)^2}{g_G^2g^2}
\le \mathbb P_s\frac{\|\bb w(T)\|^2}{\|\bb w(0)\|^2}
\le \frac{(1+\eta)^2}{g_G^2g^2}.
\end{equation}
Consequently, for $\eta$ bounded above by a fixed constant less than one, $\mathbb P_s=\Omega(g^{-2})$, while the product in \eqref{eq:clock-success-weighted} is $\Theta(g^{-2})$.
\end{proposition}

\begin{proof}
By Theorem~\ref{thm:target-node}, $g_G=\Theta(1)$, and the choice \eqref{eq:clock-success-order} gives
\[
\frac{\|\bb x_s(T)-\bb x(T)\|}{\|\bb x(T)\|}
\le C_sg\alpha_AT\e^{-c_sr}\le\eta.
\]
The transport matrix $-\i P_s\otimes I^{\otimes n}$ is anti-Hermitian, while the sampled coefficient matrix is dissipative. Hence $\bar A+\bar A^\dagger\preceq O$, and
\[
\|\bb{w}(T)\|\le\|\bb{w}(0)\|=g_G\|\bb{x}_0\|.
\]
The recovery error bound gives
\[
(1-\eta)\|\bb x(T)\|\le\|\bb x_s(T)\|\le(1+\eta)\|\bb x(T)\|.
\]
Combining the lower bound with contractivity and \eqref{eq:clock-success-exact} proves \eqref{eq:clock-success-lower}. To obtain \eqref{eq:clock-success-weighted}, we retain the exact identity
\[
\mathbb P_s\frac{\|\bb w(T)\|^2}{\|\bb w(0)\|^2}
=\frac{\|\bb x_s(T)\|^2}{g_G^2\|\bb x_0\|^2}
\]
and apply both bounds on $\|\bb x_s(T)\|$. The asymptotic statements follow from $g_G=\Theta(1)$.
\end{proof}

We use the two-sided estimate \eqref{eq:clock-success-weighted} in the full algorithm: the norm-decay factor that accompanies $\mathbb P_s$ is also present in the success probability of Schr\"odingerization. With the periodic-profile and Fourier projection errors included, the recovered autonomous vector is $\bb w_h(T)$ and the conditional clock probability becomes $\mathbb P_s^{\,h}=\|\bb x_h(T)\|^2/\|\bb w_h(T)\|^2$. In Section~\ref{sec:error-complexity}, we relate their joint recovery probability directly to \eqref{eq:clock-success-weighted}; see \eqref{eq:joint-recovery-probability}.

\section{Schr\"odingerization for autonomous non-unitary dynamics}
\label{sec:Sch-auto-complexity}

We consider the autonomous linear system
\begin{equation}
\label{eq:autonomous-system}
\frac{\d}{\d t}\bb w(t)=\bar A\bb w(t),
\qquad \bb w(0)\in\mathbb C^d\setminus\{\bb 0\},
\qquad 0\le t\le T,
\end{equation}
where $\bar A$ is time-independent and dissipative, and $d$ is a power of two. We seek $\ket{\bb w(T)}:=\bb w(T)/\|\bb w(T)\|$. The clock system is one instance of this general problem.

We apply Schr\"odingerization~\cite{JLY22SchrShort,JLY22SchrLong} to convert dissipation into transport in an auxiliary variable $p$. We refine the initialization of Ref.~\cite{JLMPY2025schr} by exact periodization and use its explicit Fourier coefficients to estimate the projection error directly. We then evolve the finite Hamiltonian system and recover $\bb w(T)$ through a weighted $L^2$ projection.

\subsection{Schr\"odingerization with periodic initialization}

We decompose the autonomous generator as
\begin{equation}
\label{eq:Sch-H1-H2-structure}
\bar{A}=H_1+\i H_2, \qquad H_1:=\frac{\bar{A}+\bar{A}^\dagger}{2}, \qquad H_2:=\frac{\bar{A}-\bar{A}^\dagger}{2\i}.
\end{equation}
We assume $H_1\preceq O$ and access to exact block-encodings of $H_1$ and $H_2$ with positive normalizations $\alpha_1$ and $\alpha_2$, respectively, including their adjoints and controlled versions. In particular, $-\alpha_1 I_d\preceq H_1\preceq O$. These assumptions cover arbitrary autonomous dissipative matrices; no clock structure or coefficient regularity is needed. For the clock system of Section~\ref{sec:auto}, the encodings can be constructed with $\alpha_1=\alpha_A$ and $\alpha_2=\alpha_A+\mu_{s,\max}$, as detailed in Section~\ref{sec:error-complexity}.

We introduce an auxiliary variable $p\in\mathbb R$ and set $\bb{q}_*(t,p):=\e^{-p}\bb{w}(t)$. Since $\partial_p\bb{q}_*=-\bb{q}_*$, we can represent the contribution of $H_1$ as transport in $p$, so \eqref{eq:autonomous-system} becomes
\begin{equation}
\label{qstar}
\partial_t\bb{q}_*(t,p) = -H_1\partial_p\bb{q}_*(t,p)+\i H_2\bb{q}_*(t,p), \qquad p \in \mathbb{R}.
\end{equation}
The ideal solution satisfies $\bb w(t)=\e^p\bb q_*(t,p)$. We use an integral form of this identity so that recovery can be implemented by a projection on the $p$-register. On the fixed interval $J=[1/2,1]$, we define
\[
\chi(p):=\frac{2\e^{-p}}{\e^{-1}-\e^{-2}}\mathbf 1_J(p), \qquad \int_J\chi(p)\e^{-p}\,\d p=1.
\]
Then
\begin{equation}
\label{eq:p-ideal-recovery}
\langle\chi,\bb{q}_*(t,\cdot)\rangle_{L^2(J)} = \int_J\chi(p)\e^{-p}\,\d p\,\bb{w}(t) = \bb{w}(t).
\end{equation}

To obtain a normalizable state on a finite periodic interval, we replace $\e^{-p}$ by a periodic profile. This profile must approximate $\e^{-p}$ on the region that determines recovery on $J$. Since $-\alpha_1 I_d\preceq H_1\preceq O$, transport is directed towards decreasing $p$ with speed at most $\alpha_1$; the term $\i H_2\bb q_*$ acts locally in $p$. The relevant initial region is therefore $[1/2,1+\alpha_1T]$. For $L\ge1$, choose
\[
\ell_p:=\alpha_1T+1+2L, \qquad \Omega_p:=[-L,\ell_p-L)=[-L,1+\alpha_1T+L).
\]
We leave a margin $L$ between this region and the periodic boundary and use the error-function profile of Ref.~\cite{JLMPY2025schr} to define
\begin{equation}
\label{eq:p-profile-periodization}
\psi_p(p):=\e^{-p}\frac{1+\erf(2\sqrt{L}\,p)}{2}, \qquad \Psi_p(p):=\sum_{m\in\mathbb Z}\psi_p(p+m\ell_p).
\end{equation}
Here
\[
\erf(x):=\frac{2}{\sqrt\pi}\int_0^x\e^{-y^2}\,\d y, \qquad \erfc(x):=1-\erf(x).
\]
The function $\psi_p$ decays in both directions: the exponential controls the positive tail, and the error-function factor suppresses the negative tail. The sum of its translates therefore converges absolutely and defines an exactly $\ell_p$-periodic profile. Increasing $L$ improves its agreement with $\e^{-p}$ on the relevant initial region.

Replacing the ideal initial profile by $\Psi_p$, we obtain the periodic system
\begin{equation}
\label{eq:Sch-warped-system}
\begin{cases} \partial_t\bb{q}(t,p) =-H_1\partial_p\bb{q}(t,p)+\i H_2\bb{q}(t,p),\\[1mm] \bb{q}(0,p)=\Psi_p(p)\bb{w}(0), \end{cases} \qquad p\in\Omega_p,
\end{equation}
with periodic boundary conditions in $p$. We need only estimate the error on $J$ to use \eqref{eq:p-ideal-recovery}; in the following lemma, we bound it independently of the transport distance $\alpha_1T$.

\begin{lemma}
\label{lem:p-recovery}
Let $\bb{q}$ solve \eqref{eq:Sch-warped-system}. For $L\ge1$,
\begin{equation}
\label{eq:local-energy}
\|\bb{q}(T,\cdot)-\bb{q}_*(T,\cdot)\|_{L^2(J)} \le C\frac{\e^{-L}}{\sqrt L}\|\bb{w}(0)\|,
\end{equation}
where $C$ is independent of $T$, $\alpha_1$, $H_2$, and the dimension $d$.
\end{lemma}
\begin{proof}
We use a local energy estimate to bound the error on $J$ at time $T$ by the initial profile error on $I_T=[1/2,1+\alpha_1T]$. Set $\bb{e}(t,p):=\bb{q}(t,p)-\bb{q}_*(t,p)$ for $p\in\Omega_p$. Subtracting \eqref{qstar} from \eqref{eq:Sch-warped-system} gives
\[
\partial_t\bb{e} =-H_1\partial_p\bb{e}+\i H_2\bb{e}.
\]
Since $H_1$ and $H_2$ are Hermitian, the $H_2$ terms cancel in the energy identity, yielding
\[
\partial_t\|\bb{e}\|^2 =-(\partial_p\bb{e})^\dagger H_1\bb{e} -\bb{e}^\dagger H_1\partial_p\bb{e} =-\partial_pj, \qquad j(t,p):=\bb{e}(t,p)^\dagger H_1\bb{e}(t,p).
\]
To control the error reaching $J$, we consider the energy on the shrinking interval
\[
E(t):=\int_{1/2}^{1+\alpha_1(T-t)} \|\bb{e}(t,p)\|^2\,\d p.
\]
This interval remains strictly inside $\Omega_p$, so the calculation does not involve the periodic boundary. Applying the Leibniz rule, we obtain
\[
\begin{aligned} E'(t) &=j(t,1/2)-j\bigl(t,1+\alpha_1(T-t)\bigr) -\alpha_1 \bigl\|\bb{e}\bigl(t,1+\alpha_1(T-t)\bigr)\bigr\|^2\\ &=\bb{e}(t,1/2)^\dagger H_1\bb{e}(t,1/2) -\bb{e}\bigl(t,1+\alpha_1(T-t)\bigr)^\dagger (H_1+\alpha_1 I_d) \bb{e}\bigl(t,1+\alpha_1(T-t)\bigr)\\ &\le0. \end{aligned}
\]
Here the last inequality uses $H_1\preceq O$ and $H_1+\alpha_1 I_d\succeq O$. Consequently,
\begin{equation}
\label{eq:p-moving-energy}
\|\bb{e}(T,\cdot)\|_{L^2(J)}^2 =E(T)\le E(0) =\|\bb{e}(0,\cdot)\|_{L^2(I_T)}^2.
\end{equation}
It remains to bound $\Psi_p-\e^{-p}$ on $I_T$. Using \eqref{eq:p-profile-periodization}, we separate this difference into the smoothing error and the two sums of shifted profiles,
\[
\Psi_p(p)-\e^{-p}=R_0(p)+R_+(p)+R_-(p),
\]
where
\[
R_0(p):=\psi_p(p)-\e^{-p}, \qquad R_+(p):=\sum_{m\ge1}\psi_p(p+m\ell_p), \qquad R_-(p):=\sum_{m\ge1}\psi_p(p-m\ell_p).
\]
To estimate $R_0$, we use the Gaussian-tail bound $1-\erf(x)\le\e^{-x^2}/(\sqrt\pi\,x)$ for $x>0$ and the condition $p\ge1/2$ to obtain
\[
|R_0(p)| \le C\frac{\e^{-4Lp^2-p}}{\sqrt L\,p} \le C\frac{\e^{-L}}{\sqrt L}\e^{-p}, \qquad \|R_0\|_{L^2(I_T)} \le C\frac{\e^{-L}}{\sqrt L}.
\]
The positive shifts are controlled by $0\le\psi_p(q)\le\e^{-q}$ for $q>0$. Summing the resulting geometric series, we obtain
\[
0\le R_+(p) \le\frac{\e^{-\ell_p}}{1-\e^{-\ell_p}}\e^{-p}, \qquad \|R_+\|_{L^2(I_T)}\le C\e^{-2L},
\]
where $\ell_p=\alpha_1T+1+2L$ and $L\ge1$. For the negative shifts, set $y_m=m\ell_p-p$. On $I_T$, one has $y_m\ge2L+(m-1)\ell_p$. Since $4Ly^2\ge2y+L$ for $L\ge1$ and $y\ge2L$, we obtain from the Gaussian-tail bound
\[
\psi_p(-y) \le C\frac{\e^{y-4Ly^2}}{\sqrt L\,y} \le C\frac{\e^{-L-y}}{\sqrt L}.
\]
Summing over $m\ge1$ and integrating on $I_T$, we obtain
\[
|R_-(p)| \le C\frac{\e^{-L}}{\sqrt L}\e^{-(\ell_p-p)}, \qquad \|R_-\|_{L^2(I_T)} \le C\frac{\e^{-3L}}{\sqrt L}.
\]
The three bounds and the triangle inequality imply $\|\Psi_p-\e^{-p}\|_{L^2(I_T)} \le C\e^{-L}/\sqrt L$. Finally, substituting $\bb{e}(0,p)=(\Psi_p(p)-\e^{-p})\bb{w}(0)$ into \eqref{eq:p-moving-energy}, we obtain
\[
\|\bb{q}(T,\cdot)-\bb{q}_*(T,\cdot)\|_{L^2(J)} \le\|\Psi_p-\e^{-p}\|_{L^2(I_T)}\|\bb{w}(0)\| \le C\frac{\e^{-L}}{\sqrt L}\|\bb{w}(0)\|.
\]
This proves \eqref{eq:local-energy}.
\end{proof}

\subsection{Fourier projection and recovery}
\label{subsec:p-fourier-recovery}

We project the periodic problem onto finitely many Fourier modes, compute the initial coefficients explicitly, and evolve the resulting Hamiltonian system. We then use Fourier inversion and weighted projection to recover the autonomous solution.

\subsubsection{Hamiltonian formulation and recovery}

To store the retained Fourier modes in an $n_p$-qubit register, let $N_p=2^{n_p}$ with $n_p\ge2$ and set $\mathcal K_p=\{0,\ldots,N_p-1\}$. Using centered frequencies, write
\[
\mu_l^p=\frac{2\pi}{\ell_p}\Bigl(l-\frac{N_p}{2}\Bigr), \qquad l\in\mathbb Z.
\]
The indices in $\mathcal K_p$ correspond to the integer modes $-N_p/2,\ldots,N_p/2-1$, with maximum absolute frequency $\mu_{p,\max}:=\pi N_p/\ell_p$. The value of $N_p$ will be fixed after Lemma~\ref{lem:matched-filter}. For an $\ell_p$-periodic function $f$, the Fourier coefficients and projection are defined by
\[
\widehat f_l=\frac{1}{\ell_p}\int_{\Omega_p} f(p)\e^{-\i\mu_l^p p}\,\d p, \qquad \Pi_pf=\sum_{l\in\mathcal K_p} \widehat f_l\e^{\i\mu_l^p p}.
\]
Since $H_1$ and $H_2$ are independent of $p$, Fourier projection commutes with the evolution. Applying $\Pi_p$ componentwise to \eqref{eq:Sch-warped-system}, we obtain
\begin{equation}
\label{eq:Fourier1}
\begin{cases} \partial_t\Pi_p\bb q =-H_1\partial_p\Pi_p\bb q+\i H_2\Pi_p\bb q,\\[1mm] (\Pi_p\bb q)(0,p)=(\Pi_p\Psi_p)(p)\bb w(0). \end{cases}
\end{equation}
Substituting
\[
(\Pi_p\bb q)(t,p)= \sum_{l\in\mathcal K_p} \widehat{\bb q}_l(t)\e^{\i\mu_l^p p}
\]
into \eqref{eq:Fourier1}, we obtain
\begin{equation}
\label{eq:p-coefficient-ODE}
\frac{\d}{\d t}\widehat{\bb q}_l(t) =-\i(\mu_l^pH_1-H_2)\widehat{\bb q}_l(t), \qquad \widehat{\bb q}_l(0)=\widehat{\Psi}_{p,l}\bb w(0).
\end{equation}
We can evaluate the initial coefficients without summing infinitely many translates. Writing $\widehat\psi_p(\xi)=\int_{\mathbb R}\psi_p(p)\e^{-\i\xi p}\,\d p$, we obtain them directly from samples of the Fourier transform:
\begin{equation}
\label{eq:p-periodic-coefficients}
\widehat{\Psi}_{p,l}
=\frac{1}{\ell_p}\int_{\mathbb R}\psi_p(p)\e^{-\i\mu_l^p p}\,\d p
=\frac{1}{\ell_p}\widehat\psi_p(\mu_l^p)
=\frac{\exp\bigl((1+\i\mu_l^p)^2/(16L)\bigr)}{\ell_p(1+\i\mu_l^p)}.
\end{equation}
We obtain the first equality by integrating the translates over one period and using $\mu_l^p\ell_p\in2\pi\mathbb Z$, and derive the last equality in the proof of Lemma~\ref{lem:p-fourier-tail}. We can therefore specify the projected profile through $N_p$ explicit coefficients, without evaluating the infinite sum in \eqref{eq:p-profile-periodization}.

Each Hermitian matrix $\mu_l^pH_1-H_2$ generates unitary evolution of its coefficient, so $\|\widehat{\bb q}_l(t)\|=\|\widehat{\bb q}_l(0)\|$. To assemble these equations into one Hamiltonian system, we define the coefficient vector
\begin{equation}
\label{eq:p-coefficient-register}
\widetilde{\bb q}(t)=\sqrt{\ell_p} \sum_{l\in\mathcal K_p} \e^{-\i\mu_l^pL}\ket l\otimes\widehat{\bb q}_l(t).
\end{equation}
We include the factor $\sqrt{\ell_p}\e^{-\i\mu_l^pL}$ to express the coefficients in the orthonormal basis $\ell_p^{-1/2}\e^{\i\mu_l^p(p+L)}$, with the origin shifted to the left endpoint of $\Omega_p$. By Parseval's identity, we have
\[
\|\widetilde{\bb q}(t)\| =\|\Pi_p\bb q(t,\cdot)\|_{L^2(\Omega_p)}.
\]

Collecting \eqref{eq:p-coefficient-ODE}, we obtain a single Hamiltonian system with $D_p=\diag(\mu_0^p,\ldots,\mu_{N_p-1}^p)$ and $\mathcal H=D_p\otimes H_1-I^{\otimes n_p}\otimes H_2$:
\begin{equation}
\label{eq:Sch-Hamiltonian-system}
\frac{\d}{\d t}\widetilde{\bb q}(t) =-\i\mathcal H\widetilde{\bb q}(t).
\end{equation}
For $l\in\mathcal K_p$, let $H_{\mu_l}=\mu_l^pH_1-H_2$. Using the block-diagonal form $\mathcal H=\sum_{l\in\mathcal K_p}\ket l\bra l\otimes H_{\mu_l}$, we obtain
\begin{equation}
\label{eq:Sch-Fourier-propagator}
\widehat{\mathcal U}(T)=\e^{-\i\mathcal HT} =\sum_{l\in\mathcal K_p}\ket l\bra l\otimes\e^{-\i H_{\mu_l}T}.
\end{equation}
Thus $\widetilde{\bb q}(T)= \widehat{\mathcal U}(T)\widetilde{\bb q}(0)$ represents the exact evolution of the Fourier-projected initial condition.

To recover the autonomous solution from this evolution, we use the coefficient vector $\widetilde{\bb q}(t)$ to represent $\Pi_p\bb q(t,\cdot)$ and apply the weighted identity \eqref{eq:p-ideal-recovery}. We define the recovered vector by
\begin{equation}
\label{eq:Sch-p-pairing}
\bb{w}_h(T) :=\langle\Pi_p\chi,\Pi_p\bb{q}(T,\cdot)\rangle_{L^2(\Omega_p)} =\int_{\Omega_p} \overline{(\Pi_p\chi)(p)}\,(\Pi_p\bb{q})(T,p)\,\d p.
\end{equation}
We extend $\chi$ by zero outside $J$ and use the orthogonality of $\Pi_p$ to obtain $\langle\Pi_p\chi,\Pi_p\bb q\rangle=\langle\chi,\Pi_p\bb q\rangle$, so this is precisely the ideal weighted recovery map applied to the projected periodic solution.

To implement this map in the grid basis, we first use a quantum Fourier transform (QFT) to convert $\widetilde{\bb q}(t)$ into grid amplitudes. Define
\[
\Delta p=\frac{\ell_p}{N_p}, \qquad p_j=-L+j\Delta p, \qquad 0\le j<N_p.
\]
With the shifted Fourier basis used in \eqref{eq:p-coefficient-register}, the required transformation is
\begin{equation}
\label{eq:p-evaluation-matrix}
(F_p)_{jl} =\frac{1}{\sqrt{N_p}}\e^{\i\mu_l^p(p_j+L)} =\frac{1}{\sqrt{N_p}}\e^{2\pi\i j(l-N_p/2)/N_p}, \qquad 0\le j,l<N_p.
\end{equation}
The centered frequencies contribute a factor $\e^{-\pi\i j}=(-1)^j$, so $F_p$ is the positive-exponent QFT followed by this known diagonal phase. Combining \eqref{eq:p-coefficient-register} and \eqref{eq:p-evaluation-matrix}, we obtain
\begin{equation}
\label{eq:p-grid-evaluation}
\sqrt{\Delta p} \sum_{j=0}^{N_p-1}\ket j\otimes(\Pi_p\bb{q})(t,p_j) =(F_p\otimes I_d)\widetilde{\bb{q}}(t).
\end{equation}
We can also express the recovery integral exactly in the grid basis. Since the retained modes are distinct modulo $N_p$, we obtain by discrete orthogonality
\begin{equation}
\label{eq:p-discrete-pairing}
\Delta p\sum_{j=0}^{N_p-1} \overline{(\Pi_pf)(p_j)}(\Pi_pg)(p_j) =\langle\Pi_pf,\Pi_pg\rangle_{L^2(\Omega_p)}
\end{equation}
for scalar $f$ and scalar- or vector-valued $g$. Taking $f=\chi$ and $g=\bb q(T,\cdot)$, we obtain
\begin{equation}
\label{eq:p-grid-recovery}
\bb w_h(T)
=\langle\Pi_p\chi,\Pi_p\bb q(T,\cdot)\rangle_{L^2(\Omega_p)}
=\Delta p\sum_{j=0}^{N_p-1}
\overline{(\Pi_p\chi)(p_j)}(\Pi_p\bb q)(T,p_j).
\end{equation}
We thus recover a state proportional to $\bb w_h(T)$ by applying a QFT and projecting onto the normalized vector proportional to $\sqrt{\Delta p}\sum_j(\Pi_p\chi)(p_j)\ket j$. These two operations add no error to the Fourier projection already made.

\subsubsection{Projection error and normalization}

We estimate the error in $\bb w_h(T)$ by separating the replacement of the ideal profile from the Fourier projection. Having bounded the first in Lemma~\ref{lem:p-recovery}, we now consider the omitted coefficients, whose norms are preserved under \eqref{eq:p-coefficient-ODE}. By Parseval's identity, we obtain
\begin{equation}
\label{eq:p-solution-tail}
\|(I-\Pi_p)\bb q(t,\cdot)\|_2 =\|(I-\Pi_p)\Psi_p\|_2\,\|\bb w(0)\|, \qquad 0\le t\le T.
\end{equation}
We therefore need only estimate the Fourier tail of the initial profile, without a time-dependent growth factor. Using the explicit coefficients in \eqref{eq:p-periodic-coefficients}, we prove the following Gaussian tail bound.

\begin{lemma}
\label{lem:p-fourier-tail}
Let $L\ge1$, let $N_p=2^{n_p}$ with $n_p\ge2$, and let $\Psi_p$ be defined by \eqref{eq:p-profile-periodization}. Then
\begin{equation}
\label{eq:p-tail}
\|(I-\Pi_p)\Psi_p\|_2 \le C\Bigl(\frac{L}{\mu_{p,\max}}\Bigr)^{1/2} \exp\Bigl(-\frac{\mu_{p,\max}^2}{64L}\Bigr).
\end{equation}
\end{lemma}

\begin{proof}
Differentiating \eqref{eq:p-profile-periodization}, we obtain
\[
\psi_p'(p)+\psi_p(p)=2\sqrt{\frac{L}{\pi}}\e^{-4Lp^2-p}.
\]
The Gaussian tail of $\erfc$ implies that $\psi_p$ decays to zero as $p\to-\infty$, while $\psi_p(p)\le\e^{-p}$ for $p\ge0$. The boundary terms therefore vanish, so we can integrate by parts to obtain
\[
(1+\i\xi)\widehat\psi_p(\xi) = \frac{2\sqrt L}{\sqrt\pi} \int_{\mathbb R} \e^{-4Lp^2-(1+\i\xi)p}\,\d p = \exp\Bigl(\frac{(1+\i\xi)^2}{16L}\Bigr).
\]
Hence,
\begin{equation}
\label{eq:p-profile-transform}
\widehat\psi_p(\xi) = \frac{\exp\bigl((1+\i\xi)^2/(16L)\bigr)}{1+\i\xi}.
\end{equation}

By \eqref{eq:p-periodic-coefficients}, \eqref{eq:p-profile-transform}, and Parseval's identity,
\[
\|(I-\Pi_p)\Psi_p\|_2^2 = \frac{1}{\ell_p} \sum_{l\notin\mathcal K_p} \frac{ \exp\bigl((1-(\mu_l^{p})^2)/(8L)\bigr) }{ 1+(\mu_l^{p})^2 }.
\]
Writing $k=l-N_p/2$, the omitted modes satisfy $k\ge N_p/2$ or $k\le-N_p/2-1$. Since $N_p\ge4$, one has $2\pi(N_p/2-1)/\ell_p\ge\mu_{p,\max}/2$. By comparison with the decreasing Gaussian integral, we obtain
\[
\begin{aligned} \frac{1}{\ell_p}\sum_{l\notin\mathcal K_p} \frac{\e^{-(\mu_l^p)^2/(8L)}}{1+(\mu_l^p)^2} &\le\frac{2}{\ell_p}\sum_{k=N_p/2}^{\infty} \e^{-(2\pi k/\ell_p)^2/(8L)}\\ &\le\frac{1}{\pi}\int_{2\pi(N_p/2-1)/\ell_p}^{\infty} \e^{-\xi^2/(8L)}\,\d\xi\\ &\le\frac{1}{\pi}\int_{\mu_{p,\max}/2}^{\infty} \e^{-\xi^2/(8L)}\,\d\xi\\ &\le C\frac{L}{\mu_{p,\max}} \e^{-\mu_{p,\max}^2/(32L)}. \end{aligned}
\]
The last step uses $\int_a^\infty\e^{-\xi^2/(8L)}\,\d\xi \le(4L/a)\e^{-a^2/(8L)}$ for $a>0$. Since $L\ge1$, the factor $\e^{1/(8L)}$ is uniformly bounded. Taking square roots proves \eqref{eq:p-tail}.

\end{proof}
In the next lemma, we combine the two error estimates and bound the normalization factors $\|\Pi_p\Psi_p\|_2$ and $\|\Pi_p\chi\|_2$.

\begin{lemma}
\label{lem:matched-filter}
Let $L\ge1$ and $N_p=2^{n_p}$ with $n_p\ge2$. There exists a sufficiently large fixed constant $b_p\ge2\pi$ such that, if $\mu_{p,\max}\ge b_pL$, then
\begin{equation}
\label{eq:p-matched-normalizations}
 \|\Pi_p\Psi_p\|_2=\Theta(1), \qquad \|\Pi_p\chi\|_2=\Theta(1),
\end{equation}
where the implied constants are independent of $L$, $\ell_p$, and $T$. Moreover,
\begin{equation}
\label{eq:matched-recovery}
\|\bb w_h(T)-\bb w(T)\| \le C\varepsilon_p\|\bb w(0)\|,
\end{equation}
where
\begin{equation}
\label{eq:rp}
\varepsilon_p := \frac{\e^{-L}}{\sqrt L} + \Bigl(\frac{L}{\mu_{p,\max}}\Bigr)^{1/2} \exp\Bigl(-\frac{\mu_{p,\max}^2}{64L}\Bigr).
\end{equation}
\end{lemma}

\begin{proof}
By \eqref{eq:p-profile-transform}, \eqref{eq:p-periodic-coefficients}, and Parseval's identity,
\[
\|\Psi_p\|_2^2 = \frac{1}{\ell_p} \sum_{l\in\mathbb Z} |\widehat\psi_p(\mu_l^{p})|^2 \le \frac{C}{\ell_p} \sum_{l\in\mathbb Z} \frac{1}{1+(\mu_l^{p})^2} \le C.
\]
To verify that the last inequality is uniform in $\ell_p$, we separate the zero frequency and compare the positive frequencies with an integral:
\[
\frac{1}{\ell_p}\sum_{l\in\mathbb Z}\frac{1}{1+(\mu_l^p)^2} \le\frac{1}{\ell_p}+\frac{1}{\pi}\int_0^\infty\frac{\d\xi}{1+\xi^2} =\frac{1}{\ell_p}+\frac12,
\]
and $\ell_p=\alpha_1T+1+2L\ge3$. Moreover, $\mu_{p,\max}\ge b_pL$ and \eqref{eq:p-tail} imply $\|(I-\Pi_p)\Psi_p\|_2^2\le Cb_p^{-1}\e^{-b_p^2L/32} \le Cb_p^{-1}\e^{-b_p^2/32}$. Thus a sufficiently large fixed $b_p$ makes this tail uniformly small. Since $\psi_p\ge0$, \eqref{eq:p-profile-periodization} gives $\Psi_p(p)\ge\psi_p(p)\ge\e^{-p}/2$ for $p\in J$, and hence $ \|\Psi_p\|_2^2\ge\frac14\int_J\e^{-2p}\,\d p. $ By the orthogonality of $\Pi_p$ and \eqref{eq:p-tail},
\[
\|\Pi_p\Psi_p\|_2^2 = \|\Psi_p\|_2^2-\|(I-\Pi_p)\Psi_p\|_2^2, \qquad \frac18\int_J\e^{-2p}\,\d p\le\|\Pi_p\Psi_p\|_2^2\le C
\]
for sufficiently large fixed $b_p$.

For the recovery function $\chi$,
\[
\widehat\chi(\xi) = \frac{2}{\e^{-1}-\e^{-2}}\frac{\e^{-(1+\i\xi)/2}-\e^{-(1+\i\xi)}}{1+\i\xi}, \qquad |\widehat\chi(\xi)| \le \frac{C}{1+|\xi|}.
\]
Using Parseval's identity and comparison with the decreasing integral, we obtain
\[
\|(I-\Pi_p)\chi\|_2^2 = \frac{1}{\ell_p} \sum_{l\notin\mathcal K_p} |\widehat\chi(\mu_l^{p})|^2 \le C\int_{\mu_{p,\max}/2}^{\infty}\frac{\d\xi}{(1+\xi)^2} \le \frac{C}{\mu_{p,\max}}.
\]
Since $\|\chi\|_2^2=2/(\e^{-1}-\e^{-2})=\Theta(1)$,
\[
\|\Pi_p\chi\|_2^2 = \|\chi\|_2^2-\|(I-\Pi_p)\chi\|_2^2, \qquad \frac12\|\chi\|_2^2\le\|\Pi_p\chi\|_2^2\le\|\chi\|_2^2
\]
after increasing $b_p$ if necessary. This proves \eqref{eq:p-matched-normalizations}.

For the vector $\bb w_h(T)$ defined in \eqref{eq:Sch-p-pairing}, orthogonality of $\Pi_p$ gives $\langle\Pi_p\chi,\Pi_p\bb q(T,\cdot)\rangle_{L^2(\Omega_p)}=\langle\chi,\Pi_p\bb q(T,\cdot)\rangle_{L^2(\Omega_p)}$. Using the support of $\chi$ in $J$ and \eqref{eq:p-ideal-recovery}, we obtain
\[
\bb w_h(T)-\bb w(T) = \langle\chi,\bb q(T,\cdot)-\bb q_*(T,\cdot)\rangle_{L^2(J)} - \langle\chi,(I-\Pi_p)\bb q(T,\cdot)\rangle_{L^2(\Omega_p)}.
\]
Hence, by Cauchy--Schwarz, \eqref{eq:local-energy}, \eqref{eq:p-solution-tail}, and Lemma~\ref{lem:p-fourier-tail},
\[
\begin{aligned} \|\bb w_h(T)-\bb w(T)\| &\le \|\chi\|_2\Bigl(\|\bb q(T,\cdot)-\bb q_*(T,\cdot)\|_{L^2(J)}+\|(I-\Pi_p)\bb q(T,\cdot)\|_2\Bigr)\\ &\le C\Bigl[\frac{\e^{-L}}{\sqrt L}+\Bigl(\frac{L}{\mu_{p,\max}}\Bigr)^{1/2}\e^{-\mu_{p,\max}^2/(64L)}\Bigr]\|\bb w(0)\|\\ &=C\varepsilon_p\|\bb w(0)\|. \end{aligned}
\]
This proves \eqref{eq:matched-recovery}.
\end{proof}

We fix $b_p$ as in Lemma~\ref{lem:matched-filter} and choose $N_p$ to be the smallest power of two satisfying $\pi N_p/\ell_p\ge b_pL$. Since $b_p\ge2\pi$, $L\ge1$, and $\ell_p\ge3$, this ensures $N_p\ge4$, and minimality gives
\[
b_pL\le\frac{\pi N_p}{\ell_p}<2b_pL.
\]
Thus the maximum Fourier frequency satisfies
\begin{equation}
\label{eq:p-frequency-scales}
\mu_{p,\max}=\max_{l\in\mathcal K_p}|\mu_l^p| =\frac{\pi N_p}{\ell_p}=\Theta(L).
\end{equation}

With this choice, we obtain exponential decay of $\varepsilon_p$ in $L$ and keep both normalization factors $\Theta(1)$. We can therefore improve the $p$-recovery accuracy by increasing $L$ without incurring a growing normalization penalty.

\subsection{Quantum implementation and success probability}
\label{sec:quantum-recovery}
We prepare the initial state, evolve it under $\mathcal H$, and project the $p$-register to recover the solution. We first describe these operations under exact evolution and then account for simulation error in the theorem below.

We assume exact preparation of $\ket{\bb w(0)}:=\bb w(0)/\|\bb w(0)\|$ and of the auxiliary states $\ket{\Psi_p}$ and $\ket{\chi_p}$ defined below, with access to the inverse preparation circuits. The query model counts each preparation or inverse as one query; its implementation in elementary gates is not included.

Since $\bb{q}(0,p)=\Psi_p(p)\bb{w}(0)$ and $\Pi_p$ acts only on the $p$-variable, the projected initial data satisfy
\begin{equation}
\label{eq:initial-grid-factorization}
(\Pi_p\bb{q})(0,p)=(\Pi_p\Psi_p)(p)\bb{w}(0).
\end{equation}
We can therefore prepare the $p$-register independently in the state
\[
\ket{\Psi_p} :=\frac{\sqrt{\Delta p}}{\|\Pi_p\Psi_p\|_2} \sum_{j=0}^{N_p-1}(\Pi_p\Psi_p)(p_j)\ket{j},
\]
whose normalization follows from \eqref{eq:p-discrete-pairing}. Using the preparation oracles, we obtain the product state
\begin{equation}
\label{eq:full-initial-state}
\ket{\Psi_0} :=\ket{\Psi_p}\otimes\ket{\bb w(0)}.
\end{equation}
Equivalently, by \eqref{eq:initial-grid-factorization},
\[
\ket{\Psi_0} =\frac{\sqrt{\Delta p}} {\|\Pi_p\Psi_p\|_2\|\bb{w}(0)\|} \sum_{j=0}^{N_p-1}\ket{j}\otimes(\Pi_p\bb{q})(0,p_j).
\]

To evolve this initial state, we apply $F_p^\dagger$ to convert it to the Fourier basis. By \eqref{eq:p-grid-evaluation},
\[
(F_p^\dagger\otimes I_d)\ket{\Psi_0} =\frac{\widetilde{\bb{q}}(0)} {\|\Pi_p\Psi_p\|_2\|\bb{w}(0)\|}.
\]
We then evolve under $\mathcal H$ and apply $F_p$ to return to the grid basis. The full unitary is
\[
\mathcal U(T) :=(F_p\otimes I_d) \widehat{\mathcal U}(T) (F_p^\dagger\otimes I_d).
\]
Using \eqref{eq:p-grid-evaluation}, we obtain
\begin{equation}
\label{eq:Sch-q-representation}
\ket{\Psi_T} :=\mathcal U(T)\ket{\Psi_0} =\frac{\sqrt{\Delta p}} {\|\Pi_p\Psi_p\|_2\|\bb{w}(0)\|} \sum_{j=0}^{N_p-1}\ket{j}\otimes(\Pi_p\bb{q})(T,p_j).
\end{equation}
We show the circuit in Fig.~\ref{fig:schr_circuit}.

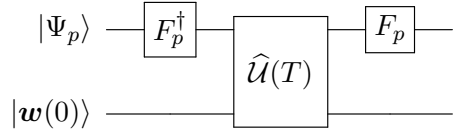
\begin{figure}[!htb]
\centering
\[
\Qcircuit @C=1.3em @R=1.5em { \lstick{\ket{\Psi_p}} & \gate{F_p^\dagger} & \multigate{1}{\widehat{\mathcal U}(T)} & \gate{F_p} & \qw \\ \lstick{\ket{\bb w(0)}} & \qw & \ghost{\widehat{\mathcal U}(T)} & \qw & \qw }
\]
\caption{Quantum circuit for the projected evolution $\mathcal U(T)$.}
\label{fig:schr_circuit}
\end{figure}

By combining the assumed block-encodings of $H_1$ and $H_2$ with a block-encoding of $D_p/\mu_{p,\max}$, we obtain
\begin{equation}
\label{eq:Sch-HAMT-Hmu}
(\bra{0^{a_H}}\otimes I_{N_pd}) \mathrm{HAM\text{-}T}_{H_\mu} (\ket{0^{a_H}}\otimes I_{N_pd}) =\sum_{l\in\mathcal K_p}\ket{l}\bra{l}\otimes \frac{H_{\mu_l}}{\alpha_1\mu_{p,\max}+\alpha_2},
\end{equation}
where $a_H$ is the number of block-encoding ancillas. Each query to $\mathrm{HAM\text{-}T}_{H_\mu}$ uses $\mathcal O(1)$ queries to the block-encodings of $H_1$ and $H_2$, including controlled and inverse queries. The normalized diagonal matrix $D_p/\mu_{p,\max}$ is specified by the known Fourier frequencies and requires no matrix-oracle queries.

After evolution, we implement the weighted recovery in \eqref{eq:p-grid-recovery} by projecting the $p$-register onto
\[
\ket{\chi_p} :=\frac{\sqrt{\Delta p}}{\|\Pi_p\chi\|_2} \sum_{j=0}^{N_p-1}(\Pi_p\chi)(p_j)\ket{j}.
\]
By \eqref{eq:Sch-q-representation} and \eqref{eq:p-discrete-pairing}, we have
\begin{equation}
\label{eq:p-projection-output}
(\bra{\chi_p}\otimes I_d)\ket{\Psi_T} =\frac{\bb{w}_h(T)} {\|\Pi_p\chi\|_2\|\Pi_p\Psi_p\|_2\|\bb{w}(0)\|}.
\end{equation}
Conditioned on success, the remaining registers encode $\bb w_h(T)/\|\bb w_h(T)\|$. This rank-one projection combines the grid amplitudes according to the recovery weight, whereas the interval measurement in Section~2.3 of \cite{JLMPY2025schr} selects all components in a recovery interval.

By \eqref{eq:p-projection-output}, the probability of this projection is
\begin{equation}
\label{eq:p-recovery-probability}
\mathbb P_p=
\frac{\|\bb w_h(T)\|^2}
{\|\Pi_p\chi\|_2^2\|\Pi_p\Psi_p\|_2^2\|\bb w(0)\|^2}.
\end{equation}
Writing $g_w:=\|\bb w(0)\|/\|\bb w(T)\|\ge1$, we obtain $\mathbb P_p=\Theta(g_w^{-2})$ from Lemma~\ref{lem:matched-filter} whenever the relative recovery error is sufficiently small.

\subsection{Query complexity of the Schr\"odingerization method}

We use the following consequence of quantum singular value transformation (QSVT)~\cite{Gilyen2019QSVD}.

\begin{lemma}
\label{lem:hamiltonian-simulation}
Let $H$ be Hermitian, $\alpha>0$, $t\in\mathbb R$, and $0<\eta<1/2$. Suppose that $U$ is an $(\alpha,m,0)$-block-encoding of $H$, with access to $U$, $U^\dagger$, and their controlled versions. Then there exists a unitary $V$ satisfying
\[
\Bigl\| (\bra{0^{m+2}}\otimes I)V(\ket{0^{m+2}}\otimes I) -\e^{-\i tH} \Bigr\| \le\eta.
\]
The circuit for $V$ uses $\mathcal O(\alpha|t|+\log(1/\eta))$ queries to these oracles.
\end{lemma}

\begin{theorem}[Autonomous dissipative evolution]
\label{thm:autonomous-Sch-complexity}
Consider \eqref{eq:autonomous-system} with $\bar A=H_1+\i H_2$, where $H_1$ and $H_2$ are Hermitian and $H_1\preceq O$. Let $T>0$, $0<\varepsilon\le1/4$, and $g_w=\|\bb w(0)\|/\|\bb w(T)\|$. Suppose exact block-encodings of $H_1$ and $H_2$ are available with positive normalizations $\alpha_1$ and $\alpha_2$, together with their adjoints and controlled versions. Assume exact preparation of $\ket{\bb w(0)}$, $\ket{\Psi_p}$, and $\ket{\chi_p}$ and access to the inverse preparation circuits.

Choose $L=\Theta(\log(g_w/\varepsilon))$ with a sufficiently large implied constant, set $\ell_p=\alpha_1T+1+2L$, and choose $N_p$ as above so that $\mu_{p,\max}=\Theta(L)$. Then the algorithm prepares a state within $\ell^2$-distance $\varepsilon$ of $\ket{\bb w(T)}$, with success probability $\Theta(1)$ and a flag indicating success, using
\begin{equation}
\label{eq:autonomous-query-complexity}
\mathcal O\!\left(
g_w\left[
\alpha_1T\log\frac{g_w}{\varepsilon}
+\alpha_2T+\log\frac{g_w}{\varepsilon}
\right]\right)
\end{equation}
queries in total to the matrix block-encodings, and $\mathcal O(g_w)$ queries to each state-preparation oracle. Before amplification, the joint probability of successful simulation-ancilla selection and $p$-recovery is $\Theta(g_w^{-2})$.
\end{theorem}

\begin{proof}
By Lemma~\ref{lem:matched-filter} and the choice $\mu_{p,\max}=\Theta(L)$, a sufficiently large constant in $L=\Theta(\log(g_w/\varepsilon))$ ensures
\[
\|\bb w_h(T)-\bb w(T)\|
\le \frac{\varepsilon}{4g_w}\|\bb w(0)\|
=\frac{\varepsilon}{4}\|\bb w(T)\|,
\qquad
\|\Pi_p\Psi_p\|_2,\|\Pi_p\chi\|_2=\Theta(1).
\]
We use Lemma~\ref{lem:hamiltonian-simulation} to simulate $\e^{-\i\mathcal HT}$ to block error $\delta$. Denote by $\mathcal U^a(T)$ the selected simulation block, conjugated by $F_p$ and $F_p^\dagger$ to return to the grid basis. Then $\|\mathcal U^a(T)-\mathcal U(T)\|\le\delta$. Define the recovered vector without normalizing the selected branch by
\[
\bb w^{\,a}(T):=
\|\Pi_p\chi\|_2\|\Pi_p\Psi_p\|_2\|\bb w(0)\|
(\bra{\chi_p}\otimes I_d)\mathcal U^a(T)\ket{\Psi_0}.
\]
Since the recovery projection is a contraction, we obtain
\[
\|\bb w^{\,a}(T)-\bb w_h(T)\|
\le \|\Pi_p\chi\|_2\|\Pi_p\Psi_p\|_2\|\bb w(0)\|\delta
\le Cg_w\delta\|\bb w(T)\|.
\]
Taking $\delta$ to be a sufficiently small constant multiple of $\varepsilon/g_w$, we obtain
\[
\|\bb w^{\,a}(T)-\bb w(T)\|
\le \|\bb w^{\,a}(T)-\bb w_h(T)\|
+\|\bb w_h(T)-\bb w(T)\|
\le\frac{\varepsilon}{2}\|\bb w(T)\|.
\]
Since $\bb w^{\,a}(T)\ne\bb 0$, we can apply the normalization estimate directly to the two final vectors to obtain
\[
\bigl\|\ket{\bb w^{\,a}(T)}-\ket{\bb w(T)}\bigr\|
\le\frac{2\|\bb w^{\,a}(T)-\bb w(T)\|}{\|\bb w(T)\|}
\le\varepsilon.
\]
The probability of selecting both the simulation ancillas and the recovery state is
\[
\frac{\|\bb w^{\,a}(T)\|^2}
{\|\Pi_p\chi\|_2^2\|\Pi_p\Psi_p\|_2^2\|\bb w(0)\|^2}
=\Theta(g_w^{-2}).
\]
We raise this probability to $\Theta(1)$ by amplitude amplification, using $\mathcal O(g_w)$ applications of the preparation circuit and its adjoint while preserving its conditional output state. By \eqref{eq:Sch-HAMT-Hmu}, the Hamiltonian normalization is $\alpha_1\mu_{p,\max}+\alpha_2$. By Lemma~\ref{lem:hamiltonian-simulation}, we therefore use
\[
\mathcal O\!\left(
T(\alpha_1\mu_{p,\max}+\alpha_2)+\log\frac1\delta
\right)
=\mathcal O\!\left(
\alpha_1T\log\frac{g_w}{\varepsilon}
+\alpha_2T+\log\frac{g_w}{\varepsilon}
\right)
\]
matrix queries per circuit application. Multiplying by the amplification cost proves \eqref{eq:autonomous-query-complexity}. Each application uses only a constant number of queries to each state-preparation oracle, giving the stated $\mathcal O(g_w)$ bound.
\end{proof}

\begin{remark}[Access through a single block-encoding]
If we are given a single exact block-encoding of $\bar A$ with normalization $\alpha$, we can construct encodings of its Hermitian components with $\alpha_1=\alpha_2=\alpha$ using a constant number of queries. Applying Theorem~\ref{thm:autonomous-Sch-complexity}, we then obtain $\mathcal O(g_w(\alpha T+1)\log(g_w/\varepsilon))$ queries to that oracle. We keep the two normalizations separate for the clock system because only $\alpha_1$ multiplies the Fourier cutoff, while the clock-transport contribution enters through $\alpha_2$.
\end{remark}

\section{Query complexity analysis of the autonomization method}
\label{sec:error-complexity}
We apply Section~\ref{sec:Sch-auto-complexity} to the clock system and prove Theorem~\ref{thm:Sch-auto-complexity}, accounting for the additional clock projection through the joint recovery probability.

Set $d=N_sN$ and take $\bar A$ from \eqref{eq:s-semidiscrete-system}. For the clock discretization, the Hermitian components are
\[
H_1= \sum_{l=0}^{N_s-1} \ket{l}\bra{l}\otimes \frac{\widetilde A(s_l)+\widetilde A(s_l)^\dagger}{2}, \qquad H_2= \sum_{l=0}^{N_s-1} \ket{l}\bra{l}\otimes \frac{\widetilde A(s_l)-\widetilde A(s_l)^\dagger}{2\i} - P_s\otimes I^{\otimes n}.
\]
By \eqref{eq:Atilde-dissipative}, $H_1\preceq O$, so this is an instance of \eqref{eq:autonomous-system}.

We encode the Hermitian matrices $(\widetilde A+\widetilde A^\dagger)/2$ and $(\widetilde A-\widetilde A^\dagger)/(2\i)$ with normalization $\alpha_A$ by taking two-term linear combinations of the oracle in \eqref{eq:HAMT-A} and its adjoint. Including $-P_s\otimes I^{\otimes n}$, we obtain block-encodings of $H_1$ and $H_2$ with normalizations $\alpha_A$ and $\alpha_A+\mu_{s,\max}$, respectively.

We can therefore use the general construction with
\[
\alpha_1=\alpha_A,\qquad
\alpha_2=\alpha_A+\mu_{s,\max},\qquad
\ell_p=\alpha_AT+1+2L.
\]
Each query to the resulting block-encoding of $\mathcal H$ uses $\mathcal O(1)$ queries to $\mathrm{HAM\text{-}T}_A$ or its adjoint.

From \eqref{eq:s-semidiscrete-system} and \eqref{eq:clock-profile-state}, we have
\[
\frac{\bb{w}(0)}{\|\bb{w}(0)\|} =\ket{G_s}\otimes\ket{\bb{x}_0}, \qquad \|\bb{w}(0)\|=g_G\|\bb{x}_0\|.
\]

Thus \eqref{eq:full-initial-state} becomes $\ket{\Psi_p}\otimes\ket{G_s}\otimes\ket{\bb x_0}$. To recover $\ket{\bb x(T)}$, we append the clock projection to the unamplified autonomous circuit and amplify the joint success event.

After $p$-recovery, we project the clock register onto $\ket{N_s/2}_s$, corresponding to $s_{N_s/2}=T$. The recovered physical vector is
\[
\bb{x}_h(T) :=(\bra{N_s/2}_s\otimes I^{\otimes n})\bb{w}_h(T).
\]
The same block of the semidiscrete solution $\bb w(T)$ is $\bb x_s(T)$. Since block selection is a contraction, we obtain from \eqref{eq:matched-recovery}
\[
\|\bb{x}_h(T)-\bb{x}_s(T)\|
\le\|\bb{w}_h(T)-\bb{w}(T)\|
\le C\varepsilon_p\|\bb{w}(0)\|.
\]
Using \eqref{eq:target-node-error}, we then bound the physical recovery error by the sum of the clock and $p$-recovery errors.

For nonzero $\bb{w}_h(T)$ and $\bb{x}_h(T)$, write $\ket{\bb{w}_h(T)}:=\bb{w}_h(T)/\|\bb{w}_h(T)\|$ and $\ket{\bb{x}_h(T)}:=\bb{x}_h(T)/\|\bb{x}_h(T)\|$. Conditioned on successful postselection, we obtain
\[
\ket{\Psi_T} \xrightarrow{\ \bra{\chi_p}\ } \ket{\bb{w}_h(T)} \xrightarrow{\ \bra{N_s/2}_s\ } \ket{\bb{x}_h(T)}.
\]
Their joint success projector is
\begin{equation}
\label{eq:success-projector}
\Pi_{\mathrm{succ}} =\ketbra{\chi_p}{\chi_p} \otimes\ket{N_s/2}_s\bra{N_s/2}_s \otimes I^{\otimes n},
\end{equation}
where the last identity acts on the physical system register. Conditioned on this outcome, the physical register is in $\ket{\bb{x}_h(T)}$.

We compute the probability of $p$-recovery from \eqref{eq:p-projection-output} and denote by $\mathbb P_s^{\,h}$ the probability of subsequently selecting the clock block from $\ket{\bb w_h(T)}$. Explicitly,
\begin{equation}
\label{eq:conditional-recovery-probabilities}
\mathbb P_p=
\frac{\|\bb{w}_h(T)\|^2}
{\|\Pi_p\chi\|_2^2\|\Pi_p\Psi_p\|_2^2\|\bb{w}(0)\|^2},
\qquad
\mathbb P_s^{\,h}=
\frac{\|\bb{x}_h(T)\|^2}{\|\bb{w}_h(T)\|^2}.
\end{equation}
The superscript $h$ distinguishes this conditional probability from $\mathbb P_s$ in \eqref{eq:clock-success-exact}: the latter concerns $\bb w(T)$ and $\bb x_s(T)$, while the former includes the $p$-profile and Fourier projection errors through $\bb w_h(T)$ and $\bb x_h(T)$.

When multiplying these probabilities, we cancel $\|\bb w_h(T)\|^2$. Using $\|\bb w(0)\|=g_G\|\bb x_0\|$ and the exact clock probability \eqref{eq:clock-success-exact}, we obtain, whenever $\bb x_s(T)\ne\bb 0$,
\begin{equation}
\label{eq:joint-recovery-probability}
\begin{aligned}
\mathbb P_{\mathrm{succ}}
&=\mathbb P_p\mathbb P_s^{\,h}
=\frac{\|\bb{x}_h(T)\|^2}
{\|\Pi_p\chi\|_2^2\|\Pi_p\Psi_p\|_2^2g_G^2\|\bb{x}_0\|^2}\\
&=\left(\mathbb P_s\frac{\|\bb w(T)\|^2}{\|\bb w(0)\|^2}\right)
\frac{\|\bb x_h(T)\|^2}{\|\bb x_s(T)\|^2}
\frac{1}{\|\Pi_p\chi\|_2^2\|\Pi_p\Psi_p\|_2^2}.
\end{aligned}
\end{equation}
The first factor in the last line is precisely the quantity bounded in Proposition~\ref{prop:clock-success}. If both $\bb x_s(T)$ and $\bb x_h(T)$ approximate $\bb x(T)$ with relative errors bounded by a fixed constant less than one, then $\|\bb x_h(T)\|/\|\bb x_s(T)\|=\Theta(1)$. Together with the bounded profile norms in Lemma~\ref{lem:matched-filter}, \eqref{eq:clock-success-weighted} therefore gives $\mathbb P_{\mathrm{succ}}=\Theta(g^{-2})$. We verify these error conditions below and include Hamiltonian simulation error.

\begin{proof}[Proof of Theorem~\ref{thm:Sch-auto-complexity}]
Let $\bb{x}_h(T)$ be the vector recovered under exact Hamiltonian evolution, and let $\bb{x}^{\,a}(T)$ denote the corresponding vector obtained by Hamiltonian simulation, as defined in part~(2). For a nonzero vector $\bb{v}$, write $\ket{\bb{v}}:=\bb{v}/\|\bb{v}\|$. The normalization estimate
\[
\Bigl\|\frac{\bb{u}}{\|\bb{u}\|} -\frac{\bb{v}}{\|\bb{v}\|}\Bigr\| \le\frac{2\|\bb{u}-\bb{v}\|}{\|\bb{v}\|}, \qquad \bb{u},\bb{v}\ne\bb{0},
\]
reduces the state error to the error between the unnormalized vectors. Applying it directly to $\bb{x}^{\,a}(T)$ and $\bb{x}(T)$, and then using the triangle inequality for the vector error, we obtain
\begin{equation}
\label{eq:Sch-normalized-recovery-error}
\begin{aligned}
\|\ket{\bb{x}^{\,a}(T)}-\ket{\bb{x}(T)}\|
&\le \frac{2\|\bb{x}^{\,a}(T)-\bb{x}(T)\|}{\|\bb{x}(T)\|}\\
&\le \frac{2}{\|\bb{x}(T)\|}
\Bigl(\|\bb{x}^{\,a}(T)-\bb{x}_h(T)\|
+\|\bb{x}_h(T)-\bb{x}(T)\|\Bigr).
\end{aligned}
\end{equation}
It therefore suffices to bound each of the two vector errors by $(\varepsilon/4)\|\bb{x}(T)\|$. These bounds also ensure that $\bb{x}^{\,a}(T)$ is nonzero, as required to apply the normalization estimate.

(1) Under exact Hamiltonian evolution, the recovery error satisfies
\begin{equation}
\label{eq:Sch-total-recovery-error}
\|\bb{x}(T)-\bb{x}_h(T)\| \le \|\bb{x}(T)-\bb{x}_s(T)\| + \|\bb{x}_s(T)-\bb{x}_h(T)\|.
\end{equation}
Choose the implied constant in $r=\Theta(\log(g\alpha_AT/\varepsilon))$ sufficiently large to meet the order requirements of Proposition~\ref{prop:clock-success}, including \eqref{eq:clock-success-order} with $\eta=\varepsilon/8$. Proposition~\ref{prop:clock-success} then gives $g_G=\Theta(1)$ and
\[
\|\bb x(T)-\bb x_s(T)\|\le\frac{\varepsilon}{8}\|\bb x(T)\|.
\]
For the $p$-recovery error, we use \eqref{eq:matched-recovery} and the contractivity of block selection to obtain
\[
\begin{aligned}
\|\bb x_s(T)-\bb x_h(T)\|
&\le\|\bb w(T)-\bb w_h(T)\|\\
&\le C\varepsilon_p\|\bb w(0)\|
=Cgg_G\varepsilon_p\|\bb x(T)\|.
\end{aligned}
\]
Since $\mu_{p,\max}\ge b_pL$, \eqref{eq:rp} gives
\[
\varepsilon_p\le C\bigl(L^{-1/2}\e^{-L}+\e^{-b_p^2L/64}\bigr).
\]
We therefore choose $L=\Theta(\log(g/\varepsilon))$, with $L\ge1$ and a sufficiently large implied constant, so that $Cgg_G\varepsilon_p\le\varepsilon/8$. Substituting the two error bounds into \eqref{eq:Sch-total-recovery-error}, we obtain
\begin{equation}
\label{eq:Sch-exact-recovery-budget}
\|\bb x_h(T)-\bb x(T)\|\le\frac{\varepsilon}{4}\|\bb x(T)\|.
\end{equation}
The reverse triangle inequality gives $\|\bb x_h(T)\|/\|\bb x_s(T)\|=\Theta(1)$, verifying the remaining condition used with Proposition~\ref{prop:clock-success} in \eqref{eq:joint-recovery-probability}. In particular, both recovered vectors are nonzero.

(2) To estimate the simulation cost, we substitute the parameters from part~(1) into \eqref{eq:s-frequency-scales} and \eqref{eq:p-frequency-scales} to obtain
\begin{equation}
\label{eq:final-cutoff-scales}
\mu_{s,\max} =\Theta\Bigl(\alpha_A+\log\frac{g\alpha_AT}{\varepsilon}\Bigr), \qquad \mu_{p,\max} =\Theta\Bigl(\log\frac{g}{\varepsilon}\Bigr).
\end{equation}
By \eqref{eq:Sch-HAMT-Hmu}, the block-encoding normalization of $\mathcal H$ is $\alpha_A\mu_{p,\max}+\alpha_A+\mu_{s,\max}$. Using \eqref{eq:final-cutoff-scales}, $\alpha_A\ge1$, and $\alpha_AT\ge1$, we obtain
\begin{equation}
\label{eq:alphaHT}
T\bigl(\alpha_A\mu_{p,\max}+\alpha_A+\mu_{s,\max}\bigr) =\mathcal O\Bigl( \alpha_AT\log\frac{g\alpha_AT}{\varepsilon} \Bigr).
\end{equation}

We apply Lemma~\ref{lem:hamiltonian-simulation} with accuracy $\delta$ and let $V(T)$ be the resulting unitary. Selecting the all-zero simulation ancillas and returning to the grid representation defines
\[
\mathcal U^a(T) := (F_p\otimes I^{\otimes(n_s+n)}) (\bra{0^{a_H+2}}\otimes I^{\otimes(n_p+n_s+n)}) V(T) (\ket{0^{a_H+2}}\otimes I^{\otimes(n_p+n_s+n)}) (F_p^\dagger\otimes I^{\otimes(n_s+n)}).
\]
Using the block-encoding error bound in Lemma~\ref{lem:hamiltonian-simulation} and unitary invariance of the operator norm, we obtain
\begin{equation}
\label{eq:simulation-error}
\|\mathcal U^a(T)-\mathcal U(T)\|\le\delta.
\end{equation}
Applying $p$-recovery and the clock projection from Section~\ref{sec:error-complexity}, we define
\[
\bb{x}^{\,a}(T) := \|\Pi_p\chi\|_2\|\Pi_p\Psi_p\|_2\|\bb{w}(0)\| (\bra{\chi_p}\otimes\bra{N_s/2}_s\otimes I^{\otimes n}) \mathcal U^a(T)\ket{\Psi_0}.
\]
Since the recovery projection has operator norm one, \eqref{eq:simulation-error} implies
\[
\|\bb{x}^{\,a}(T)-\bb{x}_h(T)\| \le \|\Pi_p\chi\|_2\|\Pi_p\Psi_p\|_2\|\bb{w}(0)\|\delta.
\]
Using \eqref{eq:p-matched-normalizations}, $\|\bb{w}(0)\|=g_G\|\bb{x}_0\|$, $g_G=\Theta(1)$, and $g=\|\bb{x}_0\|/\|\bb{x}(T)\|$, we obtain
\[
\|\bb{x}^{\,a}(T)-\bb{x}_h(T)\| \le Cg\delta\|\bb{x}(T)\|.
\]
We choose $\delta$ to be a sufficiently small constant multiple of $\varepsilon/g$. Then $\|\bb{x}^{\,a}(T)-\bb{x}_h(T)\| \le(\varepsilon/4)\|\bb{x}(T)\|$. Combining this with \eqref{eq:Sch-exact-recovery-budget}, we obtain
\[
\|\bb{x}^{\,a}(T)-\bb{x}(T)\|
\le \|\bb{x}^{\,a}(T)-\bb{x}_h(T)\|
+\|\bb{x}_h(T)-\bb{x}(T)\|
\le\frac{\varepsilon}{2}\|\bb{x}(T)\|.
\]
In particular, $\bb{x}^{\,a}(T)\ne\bb{0}$, so we can apply \eqref{eq:Sch-normalized-recovery-error} to obtain
\[
\|\ket{\bb{x}^{\,a}(T)}-\ket{\bb{x}(T)}\| \le\varepsilon.
\]

(3) We include simulation-ancilla selection in the joint recovery probability \eqref{eq:joint-recovery-probability}. Parts~(1) and~(2) give $\|\bb{x}^{\,a}(T)\|=\Theta(\|\bb{x}(T)\|)$. Since $\mathcal U^a(T)$ includes selection of the all-zero simulation ancillas without renormalization, \eqref{eq:success-projector} and the definition of $\bb{x}^{\,a}(T)$ give
\[
\begin{aligned} \mathbb P_{\mathrm{succ}}^a &=\|\Pi_{\mathrm{succ}}\mathcal U^a(T)\ket{\Psi_0}\|^2\\ &=\frac{\|\bb{x}^{\,a}(T)\|^2} {\|\Pi_p\chi\|_2^2\|\Pi_p\Psi_p\|_2^2\|\bb{w}(0)\|^2}\\ &=\Theta(g^{-2}). \end{aligned}
\]
The last equality follows by replacing $\bb x_h(T)$ with $\bb x^{\,a}(T)$ in \eqref{eq:joint-recovery-probability} and applying Proposition~\ref{prop:clock-success}: parts~(1)--(2) give $\|\bb x^{\,a}(T)\|/\|\bb x_s(T)\|=\Theta(1)$, while the profile norms are bounded by \eqref{eq:p-matched-normalizations}. We apply amplitude amplification to the implemented circuit and its adjoint to raise the success probability to $\Theta(1)$ with $\mathcal O(g)$ applications, while preserving the conditional output state $\ket{\bb{x}^{\,a}(T)}$. For $\delta=\Theta(\varepsilon/g)$, Lemma~\ref{lem:hamiltonian-simulation} and \eqref{eq:alphaHT} give $\mathcal O(\alpha_AT\log(g\alpha_AT/\varepsilon))$ queries per simulation. Each query to the block-encoding of $\mathcal H$ uses $\mathcal O(1)$ queries to $\mathrm{HAM\text{-}T}_A$ or its adjoint. Multiplying by the amplification factor gives
\[
\mathcal O\Bigl( g\alpha_AT\log\frac{g\alpha_AT}{\varepsilon} \Bigr)
\]
queries to $\mathrm{HAM\text{-}T}_A$ or its adjoint, together with $\mathcal O(g)$ queries to each state-preparation oracle.
\end{proof}

\section{Extension to inhomogeneous systems}
\label{sec:inhomogeneous}

We consider
\begin{equation}
\label{eq:inhomogeneous-ode}
\dot{\bb x}(t)=A(t)\bb x(t)+\bb b(t),
\qquad \bb x(0)=\bb x_0,
\end{equation}
and keep source-state preparation separate from matrix queries. For time-dependent sources, we combine Schr\"odingerization with Duhamel quadrature and LCU. For constant sources, we obtain a homogeneous reduction that avoids source quadrature and gives Corollary~\ref{cor:constant-source-complexity}.

\subsection{Time-dependent sources}

Using the clock profile $G$ and coefficient extension from Section~\ref{sec:auto}, we set
\begin{equation}
\label{eq:inhomogeneous-clock-pde}
\begin{cases}
\partial_t\bb z(t,s)+\partial_s\bb z(t,s)
=\widetilde A(s)\bb z(t,s)+G(s-t)\bb b(t),\\[1mm]
\bb z(0,s)=G(s)\bb x_0.
\end{cases}
\end{equation}
The generator is autonomous, while the source remains time-dependent and follows the moving clock support. For continuous $\bb b$ and $a\in[-R_s,R_s]$, variation of constants along $s=a+t$ gives
\begin{equation}
\label{eq:inhomogeneous-characteristic}
\bb z(t,a+t)=G(a)\left[
U_A(a+t,a)\bb x_0
+\int_0^t U_A(a+t,a+\tau)\bb b(\tau)\,\d\tau
\right].
\end{equation}
The solution vanishes outside $[t-R_s,t+R_s]$. Taking $a=0$ and using $G(0)=1$, we recover $\bb z(t,t)=\bb x(t)$.

By Fourier collocation on the same clock grid, we obtain
\begin{equation}
\label{eq:inhomogeneous-clock-grid}
\dot{\bb w}(t)=\bar A\bb w(t)+\bb f_h(t),
\qquad
\bb f_h(t)=\bb G_h(t)\otimes\bb b(t),
\qquad
\bb G_h(t):=\sum_{j=0}^{N_s-1}G(s_j-t)\ket j,
\end{equation}
with $\bb w(0)=\bb G_h(0)\otimes\bb x_0$ and the source-independent matrix $\bar A$ in \eqref{autoA}. We recover the physical block at $s_{N_s/2}=T$. Periodic translation preserves the mass and derivative norms of $G$, so the proof of Lemma~\ref{lem:discrete-clock-normalization} gives, uniformly in $t$,
\begin{equation}
\label{eq:translated-clock-normalization}
g_G(t):=\|\bb G_h(t)\|=\Theta(1),
\qquad 0\le t\le T.
\end{equation}
We therefore prepare the normalized source as $\bb G_h(t)/g_G(t)\otimes\ket{\bb b(t)}$. The characteristic formula also applies to $\bb b\in L^1([0,T];\mathbb C^N)$, with the equation holding almost everywhere along characteristics. Since
\[
\partial_s^j\bigl(G(s-t)\bb b(t)\bigr)
=G^{(j)}(s-t)\bb b(t),
\]
clock differentiation requires no time derivatives of $\bb b$. Time-quadrature estimates require separate assumptions, such as Lipschitz continuity for low-order rules or finite-order derivative bounds for high-order rules.

\paragraph{Duhamel quadrature.}
We write
\begin{equation}
\label{eq:inhomogeneous-grid-duhamel}
\bb w(T)=\e^{T\bar A}\bb w(0)
+\int_0^T\e^{(T-\tau)\bar A}\bb f_h(\tau)\,\d\tau.
\end{equation}
Choose quadrature nodes $\tau_k\in[0,T]$ and positive weights $\omega_k$, $1\le k\le m$. To include the initial contribution in the same sum, set
\[
\begin{aligned}
&a_0=\|\bb w(0)\|,\qquad
\bb v_0=\bb w(0)/a_0,\qquad t_0=T,\\
&a_k=\omega_k\|\bb f_h(\tau_k)\|,\qquad
\bb v_k=\bb f_h(\tau_k)/\|\bb f_h(\tau_k)\|,
\qquad t_k=T-\tau_k.
\end{aligned}
\]
We omit zero-weight terms. Define
\begin{equation}
\label{eq:inhomogeneous-quadrature}
B_h:=\sum_{k=0}^m a_k,
\qquad
\bb w_Q(T):=\sum_{k=0}^m a_k\e^{t_k\bar A}\bb v_k.
\end{equation}
We assume coherent preparation of the source states, translated clock profiles, and quadrature weights,
\[
V_a\ket0=\sum_{k=0}^m\sqrt{a_k/B_h}\ket k,
\]
with controlled and inverse access. Weight preparation requires source-norm information; its construction cost is not included in the oracle-query count.

We prepare $\sum_k\sqrt{a_k/B_h}\ket k\otimes\bb v_k$ and a common $p$-profile, then apply controlled Schr\"odingerization evolution for durations $t_k$. We choose $L$, $\ell_p$, and $N_p$ uniformly for $0\le t_k\le T$. Applying $V_a^\dagger$ and selecting zero implements the LCU sum~\cite{ACL2023LCH2}; subsequent $p$ and clock recovery select the physical output. We preserve the unnormalized amplitudes and amplify the joint success event.

Let $\bb x^{\,a}(T)$ denote the implemented unnormalized output, including discretization and simulation errors. As in \eqref{eq:p-projection-output}, the selected branch is
\begin{equation}
\label{eq:inhomogeneous-selected-branch}
\frac{\bb x^{\,a}(T)}
{B_h\|\Pi_p\Psi_p\|_2\|\Pi_p\chi\|_2}.
\end{equation}
The joint success probability is therefore
\begin{equation}
\label{eq:inhomogeneous-success}
\mathbb P_{\mathrm{succ}}
=\frac{\|\bb x^{\,a}(T)\|^2}
{B_h^2\|\Pi_p\Psi_p\|_2^2\|\Pi_p\chi\|_2^2}
=\Theta\!\left(\frac{\|\bb x(T)\|^2}{B_h^2}\right),
\end{equation}
provided $\bb x(T)\ne\bb0$, the profile norms satisfy Lemma~\ref{lem:matched-filter}, and the relative output error is bounded by a fixed constant less than one.

By \eqref{eq:translated-clock-normalization}, we have
\[
B_h=\Theta\!\left(\|\bb x_0\|+
\sum_{k=1}^m\omega_k\|\bb b(\tau_k)\|\right).
\]
If the scalar quadrature approximates $\int_0^T\|\bb b(t)\|\,\d t$ within fixed multiplicative constants, amplification costs $\Theta(g_b)$, where
\begin{equation}
\label{eq:inhomogeneous-amplification}
g_b:=\frac{\|\bb x_0\|+\int_0^T\|\bb b(t)\|\,\d t}
{\|\bb x(T)\|}.
\end{equation}
Each circuit application uses $\mathcal O(1)$ source-state queries. We implement the controlled durations by simulating $\sum_k(t_k/T)\ket k\bra k\otimes\mathcal H$ for time $T$, with normalization at most $\alpha_A\mu_{p,\max}+\alpha_A+\mu_{s,\max}$. Lemma~\ref{lem:hamiltonian-simulation} gives a per-attempt matrix-query cost
\[
\mathcal O\!\left(T(\alpha_A\mu_{p,\max}+\alpha_A+\mu_{s,\max})+\log\frac1\delta\right).
\]
We choose $\delta=\Theta(\varepsilon/g_b)$ because its vector-error contribution is at most $C B_h\delta$, and take $L,\mu_{p,\max}=\Theta(\log(g_b/\varepsilon))$ to control $p$-recovery uniformly over the quadrature branches. The conditional query bounds become
\begin{equation}
\label{eq:inhomogeneous-query-separation}
Q_b=\mathcal O(g_b),
\qquad
Q_A=\mathcal O\!\left(g_b\left[
\alpha_AT\log\frac{g_b}{\varepsilon}
+T\mu_{s,\max}+\log\frac{g_b}{\varepsilon}
\right]\right).
\end{equation}
Each other preparation oracle is also called $\mathcal O(g_b)$ times. Thus the simulation and $p$-recovery costs have single-logarithmic precision dependence. If the forced clock error admits $\mu_{s,\max}=\mathcal O(\alpha_A+\log(g_b\alpha_AT/\varepsilon))$, then for $\alpha_A\ge1$ and $\alpha_AT\ge1$ the matrix-query bound reduces to $\mathcal O(g_b\alpha_AT\log(g_b\alpha_AT/\varepsilon))$. Establishing this clock cutoff and controlling time-quadrature error require further analysis under specified source regularity and preparation assumptions.

\subsection{Constant sources}
\label{subsec:constant-source}

For $\bb b(t)\equiv\bb b$, we augment the original ODE before introducing the clock, using the scaled construction in Ref.~\cite{JLMY2025Preconditioning}. For $T>0$, set
\begin{equation}
\label{eq:constant-source-augmentation}
\bb y(t):=\begin{bmatrix}\bb x(t)\\ T\bb b\end{bmatrix},
\qquad
B(t):=\begin{bmatrix}A(t)&I/T\\O&O\end{bmatrix}.
\end{equation}
Then \eqref{eq:inhomogeneous-ode} is equivalent to
\[
\dot{\bb y}(t)=B(t)\bb y(t),
\qquad
\bb y(0)=\begin{bmatrix}\bb x_0\\T\bb b\end{bmatrix}.
\]
The vector $\bb b$ appears only in the initial state, so block-encoding $B(t)$ requires no source-state queries. The augmentation adds one data qubit and avoids source quadrature, even when $A(t)$ depends on time.

\paragraph{A dissipative reduction.}
To restore dissipativity, we set
\begin{equation}
\label{eq:constant-source-dissipative-shift}
\bb v(t):=\e^{-t/(2T)}\bb y(t),
\qquad
\widehat A(t):=B(t)-\frac{1}{2T}I_{2N},
\qquad
\dot{\bb v}(t)=\widehat A(t)\bb v(t).
\end{equation}
Writing $H_A(t)=(A(t)+A(t)^\dagger)/2\preceq O$, we have
\begin{equation}
\label{eq:constant-source-dissipativity}
\frac{\widehat A(t)+\widehat A(t)^\dagger}{2}
=\begin{bmatrix}H_A(t)&O\\O&O\end{bmatrix}
-\frac{1}{2T}\begin{bmatrix}I&-I\\-I&I\end{bmatrix}
\preceq O.
\end{equation}
We apply the same transformation on $\Omega_s'$ and then zero-extend $\widehat A$, allowing us to use the homogeneous clock construction. The $1/T$ scaling limits the accumulated damping to a constant:
\begin{equation}
\label{eq:constant-source-augmented-recovery}
\bb v(T)=\e^{-1/2}
\left(\ket0\otimes\bb x(T)+\ket1\otimes T\bb b\right).
\end{equation}
After clock and $p$ recovery, selecting the augmentation qubit in $\ket0$ yields $\ket{\bb x(T)}$.

\paragraph{Preparation and recovery.}
The normalized initial vector is
\begin{equation}
\label{eq:constant-source-initial-state}
\ket{\bb v(0)}=
\frac{\|\bb x_0\|\ket0\otimes\ket{\bb x_0}
+T\|\bb b\|\ket1\otimes\ket{\bb b}}
{\sqrt{\|\bb x_0\|^2+T^2\|\bb b\|^2}}.
\end{equation}
Given the norms, we prepare this direct-sum state with a flag rotation and $\mathcal O(1)$ controlled calls to each preparation oracle, omitting zero-weight branches.

For the exact augmented state, the conditional probability of selecting the physical block is
\begin{equation}
\label{eq:constant-source-block-probability}
\mathbb P_{\mathrm{phys}}
=\frac{\|\bb x(T)\|^2}
{\|\bb x(T)\|^2+T^2\|\bb b\|^2}.
\end{equation}
Multiplying by the clock and $p$ recovery probability cancels the augmented final norm. As in Section~\ref{sec:error-complexity}, with errors small relative to the physical output, we obtain
\begin{equation}
\label{eq:constant-source-joint-success}
\mathbb P_{\mathrm{succ}}
=\Theta\!\left(
\frac{\e^{-1}\|\bb x(T)\|^2}
{\|\bb x_0\|^2+T^2\|\bb b\|^2}
\right)=\Theta(g_b^{-2}),
\end{equation}
where
\begin{equation}
\label{eq:constant-source-normalization}
g_b=\frac{\|\bb x_0\|+T\|\bb b\|}{\|\bb x(T)\|},
\qquad
\|\bb v(0)\|=\sqrt{\|\bb x_0\|^2+T^2\|\bb b\|^2}
=\Theta\!\left(\|\bb x_0\|+T\|\bb b\|\right).
\end{equation}
\begin{corollary}[Constant inhomogeneous source]
\label{cor:constant-source-complexity}
Consider \eqref{eq:inhomogeneous-ode} with $\bb b(t)\equiv\bb b$ and $\bb x(T)\ne\bb0$. Assume the dissipativity, sampled matrix access, and finite-order coefficient bounds of Theorem~\ref{thm:Sch-auto-complexity}, with the order chosen as a sufficiently large integer
\[
r=\Theta\!\left(\log\frac{g_b\alpha_AT}{\varepsilon}\right).
\]
Suppose $\alpha_AT\ge1$, $0<\varepsilon\le1/4$, the nonzero initial and source states can be prepared exactly with controlled and inverse access, their norms are known, and the auxiliary profile states can be prepared as in the homogeneous algorithm. Then the augmented construction prepares a state within distance $\varepsilon$ of $\ket{\bb x(T)}$, with success probability $\Theta(1)$, using
\begin{equation}
\label{eq:constant-source-query-complexity}
Q_A=\mathcal O\!\left(g_b\alpha_AT\log\frac{g_b\alpha_AT}{\varepsilon}\right),
\qquad Q_b=\mathcal O(g_b).
\end{equation}
The other state-preparation oracles are also called $\mathcal O(g_b)$ times. The construction requires no source-time quadrature or LCU sum over injection times.
\end{corollary}

\begin{proof}
By \eqref{eq:constant-source-dissipativity}, $\widehat A$ is dissipative. We block-encode it with $\mathcal O(1)$ matrix queries and normalization $\alpha_A+3/(2T)=\mathcal O(\alpha_A)$, using $\alpha_AT\ge1$. Its derivatives satisfy $\widehat A^{(j)}=\diag(A^{(j)},O)$ for $j\ge1$, while the zeroth-order bound increases by at most $3/(2T)$. Thus the finite-order assumptions remain valid with a constant-factor increase in normalization.

We apply the construction and error estimates in the proof of Theorem~\ref{thm:Sch-auto-complexity} to $\bb v$. Since
\[
\frac{\|\bb v(0)\|}{\e^{-1/2}\|\bb x(T)\|}=\Theta(g_b),
\]
we can choose $r=\Theta(\log(g_b\alpha_AT/\varepsilon))$, $L=\Theta(\log(g_b/\varepsilon))$, and simulation block error $\delta=\Theta(\varepsilon/g_b)$, with suitable constants, so that the unnormalized recovered augmented vector satisfies
\[
\|\bb v^{\,a}(T)-\bb v(T)\|
\le\frac{\varepsilon}{2}\e^{-1/2}\|\bb x(T)\|.
\]
Physical-block selection is a contraction, so normalization gives state error at most $\varepsilon$ and validates \eqref{eq:constant-source-joint-success}. Each attempt uses $\mathcal O(\alpha_AT\log(g_b\alpha_AT/\varepsilon))$ matrix queries and $\mathcal O(1)$ preparation queries. Joint amplitude amplification costs $\mathcal O(g_b)$ attempts.
\end{proof}

\begin{remark}[Order of augmentation and autonomization]
Augmenting before clock discretization lets the concentrated profile multiply the entire initial vector $(\bb x_0,T\bb b)$, preserving its bounded clock normalization. For a time-dependent source, the auxiliary vector has derivative $T\dot{\bb b}(t)$, so this homogeneous reduction requires additional structure.
\end{remark}

\section{Numerical experiment}
\label{sec:numerical}
We illustrate the construction on a time-modulated two-cavity system, first without a source and then with time-dependent forcing. Using the UnitaryLab statevector simulator~\cite{UnitaryLabSoftware}, we compare the recovered solutions with direct ODE integration and examine the errors and recovery probabilities.

\subsection{Homogeneous dynamics}
\label{subsec:numerical-homogeneous}

\paragraph{Model and discretization.}
We consider a two-cavity model with periodically modulated resonance frequencies, as in~\cite{Mock2020CoupledCavity}. In a rotating frame, let $\bb{x}(t)=[a_1(t),a_2(t)]^\top\in\mathbb C^2$ denote the cavity amplitudes and solve
\begin{equation}
\label{eq:numerical-Schrodingerization-system}
\frac{\d}{\d t}\bb{x}(t)=A(t)\bb{x}(t), \qquad \bb{x}(0)=\begin{bmatrix}1\\0\end{bmatrix}, \qquad 0\le t\le1,
\end{equation}
with
\[
A(t)=-\gamma_{\mathrm c}I-\i B(t), \qquad
B(t)=\begin{bmatrix}\delta\cos(2\pi t)&\kappa\\ \kappa&-\delta\cos(2\pi t)\end{bmatrix}.
\]
We take the loss rate $\gamma_{\mathrm c}=0.3$, coupling strength $\kappa=2$, and modulation depth $\delta=0.8$. Since $(A(t)+A(t)^\dagger)/2=-\gamma_{\mathrm c}I$, we have the exact identities
\[
|a_1(t)|^2+|a_2(t)|^2=\e^{-0.6t}, \qquad g=\e^{0.3}\approx1.34986.
\]
We use the same analytic expression to extend $A$ to $\Omega_s'=[-1,2]$ and set it to zero on the rest of $\Omega_s=[-1,3)$, following \eqref{eq:A-zero-extension}.

Unless we vary a parameter explicitly, we use $T=R_s=1$, $r=12$, $M_s=16$, and $N_s=128$ for the clock profile and grid. In the $p$-direction, we set $L=4$, $\alpha_A=3.3$, and $N_p=128$, giving $\ell_p=12.3$. This choice of $\alpha_A$ bounds $\|A(s)\|\le\sqrt{\gamma_{\mathrm c}^2+\kappa^2+\delta^2}<3.3$. We evaluate the projected periodic profile from \eqref{eq:p-periodic-coefficients} and compute the Fourier coefficients of $\chi$ by integrating its exponential expression over $[1/2,1]$. The circuit uses 15 qubits: seven for each auxiliary register and one for the cavity amplitudes.

\paragraph{Circuit implementation and reference solutions.}
We run UnitaryLab version 1.1.6 with its NumPy CPU backend in double precision. To simulate the finite Hamiltonian, we write
\[
\mathcal H=\mathcal K+\mathcal V, \qquad
\mathcal K=I_{N_p}\otimes P_s\otimes I_2,
\]
where $\mathcal V$ is block diagonal in the $p$-frequency and clock-grid indices, with blocks
\[
\mathcal V_{lj}=\mathbf1_{\Omega_s'}(s_j)
\bigl(B(s_j)-\gamma_{\mathrm c}\mu_l^p I_2\bigr).
\]
For $m$ steps of size $h=1/m$, we use the symmetric product formula
\[
\e^{-\i\mathcal H}\approx
\left(\e^{-\i\mathcal K h/2}\e^{-\i\mathcal Vh}\e^{-\i\mathcal K h/2}\right)^m.
\]
We implement clock transport with QFTs and phase gates, and the blocks of $\mathcal V$ with controlled phase gates and controlled single-qubit unitaries. After evolution, we apply the $p$-register QFT and the inverse preparation of $\ket{\chi_p}$, then select the zero outcome in this register and the clock node $s=t$. We recover the unnormalized amplitudes by restoring the known initialization and recovery normalization factors.

We sample the trajectory at $t_j=j/16$, $0\le j\le16$, which are clock nodes, and compute the classical reference solution using an adaptive eighth-order Runge--Kutta method (DOP853), implemented in SciPy, with relative tolerance $2\times10^{-13}$ and absolute tolerance $2\times10^{-14}$. To distinguish the errors, we also compute $\bb x_s$ by matrix exponentiation of the autonomous clock system and $\bb x_h$ by Hermitian eigendecomposition of each retained $p$-mode Hamiltonian. These classical calculations provide independent references for the discretized circuit. We use the product formula for this numerical implementation; the query bound in Theorem~\ref{thm:Sch-auto-complexity} instead uses the Hamiltonian simulation primitive of Lemma~\ref{lem:hamiltonian-simulation}. Accordingly, this experiment tests the construction and recovery, rather than the asymptotic query bound.

\paragraph{Accuracy and refinement.}
We first refine the clock discretization while keeping $N_s/(2M_s)=4$. In Table~\ref{tab:unitarylab-clock-refinement}, we report the final-time relative clock error and the sampled profile norm. We obtain rapid error reduction until the error reaches the scale of the reference accuracy, while $g_G^2$ remains between $2.66$ and $3.67$ on these grids.

\begin{table}[htb]
\centering
\caption{Clock refinement with exact autonomous matrix evolution.}
\label{tab:unitarylab-clock-refinement}
\begin{tabular}{rrrrr}
\hline
$r$ & $N_s$ & $M_s$ & $g_G^2$ & $\|\bb x_s(1)-\bb x(1)\|/\|\bb x(1)\|$\\
\hline
4 & 32 & 4 & 2.66354 & $1.60\times10^{-4}$\\
8 & 64 & 8 & 3.05784 & $9.80\times10^{-9}$\\
12 & 128 & 16 & 3.41932 & $1.15\times10^{-14}$\\
16 & 256 & 32 & 3.66497 & $1.24\times10^{-14}$\\
\hline
\end{tabular}
\end{table}

With the default grids, we next run the circuit with $m=16,32,64,128,256$. In Fig.~\ref{fig:Schrodingerization-validation}\subref{fig:unitarylab-convergence}, we compare both the relative vector error $\|\bb x^{\,a}(1)-\bb x(1)\|/\|\bb x(1)\|$ and the normalized state error $\|\ket{\bb x^{\,a}(1)}-\ket{\bb x(1)}\|$. The latter decreases from $1.62\times10^{-3}$ to $6.28\times10^{-6}$, approximately by a factor of four whenever we double the number of steps, as expected for the second-order product formula. At 256 steps, the final relative vector error is $7.57\times10^{-6}$. We show the recovered intracavity energies in panel~\subref{fig:unitarylab-energies}.

To isolate the profile and Fourier projection errors, we use exact finite-Hamiltonian evolution and vary $(L,N_p)$ over $(2,64)$, $(4,128)$, $(6,256)$, and $(8,256)$, keeping the default clock parameters. As shown in panel~\subref{fig:unitarylab-profile}, the maximum relative vector error over the sampled times decreases from $5.48\times10^{-3}$ to $2.68\times10^{-6}$. For $L=4$, this maximum is $3.48\times10^{-4}$, attained at $t=0$, whereas the final-time error is $4.23\times10^{-6}$. The initial discrepancy arises because we apply the approximate periodic-profile recovery even at $t=0$. Thus, time-step refinement alone cannot remove the remaining profile error.

\paragraph{Recovery probability.}
We compute the probabilities directly from the simulated statevector, without finite-shot sampling or amplitude amplification. For the 256-step circuit at $T=1$, we obtain
\[
\mathbb P_p^{\,a}\approx0.155769, \qquad
\mathbb P_s^{\,a}\approx0.292455, \qquad
\mathbb P_{\mathrm{succ}}^{\,a}
=\mathbb P_p^{\,a}\mathbb P_s^{\,a}\approx0.0455555,
\]
where the superscript $a$ indicates the product-formula implementation. We show their time dependence in panel~\subref{fig:unitarylab-probabilities}. In this example, we find $\mathbb P_s^{\,a}\approx g_G^{-2}$ throughout the sampled trajectory. The uniform physical loss accounts for the decay of the joint probability, up to the profile and discretization errors; at the final time, we find
\[
\mathbb P_{\mathrm{succ}}^{\,a}\approx
\frac{g^{-2}}{\|\Pi_p\chi\|_2^2\|\Pi_p\Psi_p\|_2^2g_G^2},
\]
with $\|\Pi_p\chi\|_2^2\approx8.24775$ and $\|\Pi_p\Psi_p\|_2^2\approx0.427180$. This agreement illustrates the cancellation of the intermediate norm in \eqref{eq:joint-recovery-probability}.

\begin{figure}[!htb]
\centering
\subfigure[Recovered dynamics (256 steps).]{%
\includegraphics[width=0.48\textwidth]{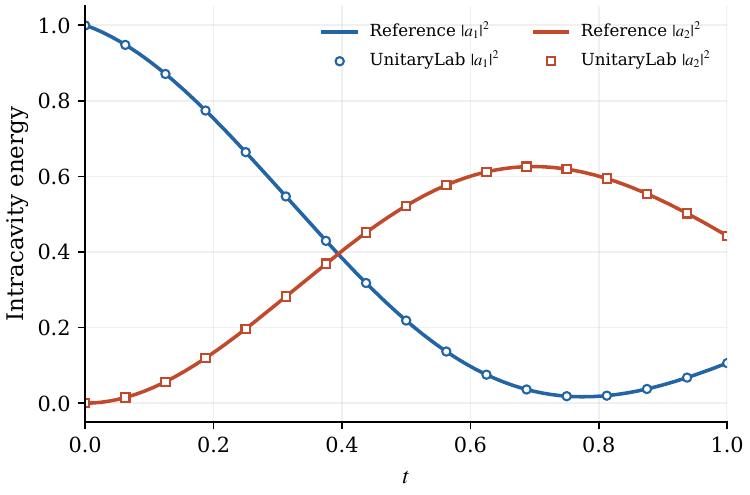}%
\label{fig:unitarylab-energies}}
\hfill
\subfigure[Circuit convergence.]{%
\includegraphics[width=0.48\textwidth]{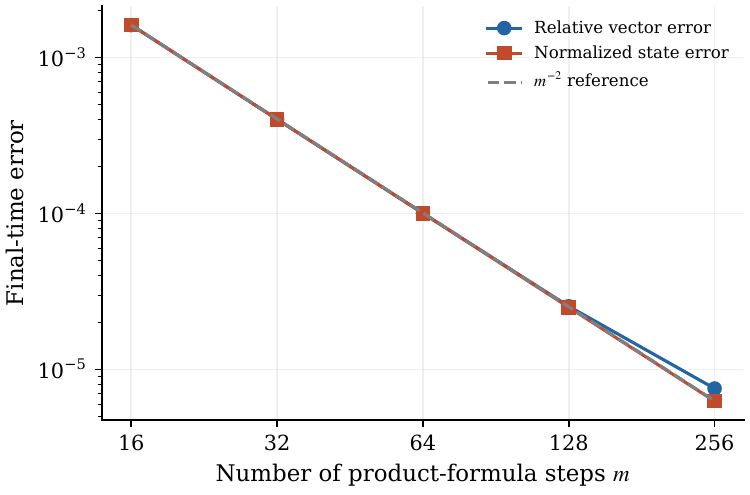}%
\label{fig:unitarylab-convergence}}
\par\medskip
\subfigure[Profile refinement (exact matrix evolution).]{%
\includegraphics[width=0.48\textwidth]{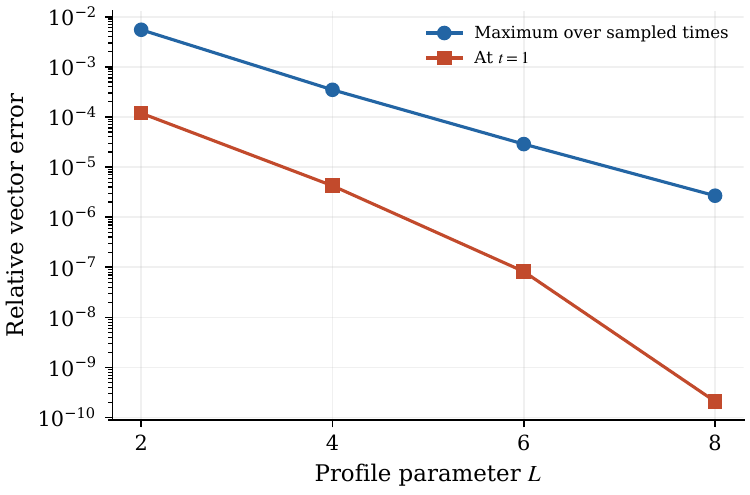}%
\label{fig:unitarylab-profile}}
\hfill
\subfigure[Recovery before amplification.]{%
\includegraphics[width=0.48\textwidth]{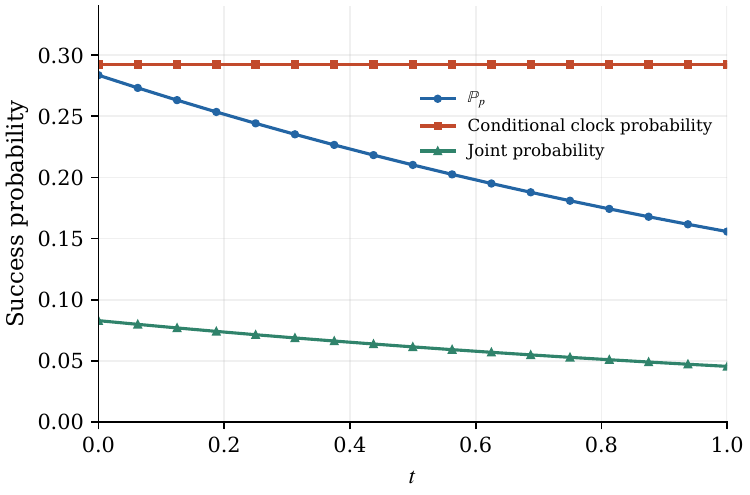}%
\label{fig:unitarylab-probabilities}}
\caption{Two-cavity experiment. (a) Intracavity energies from the 256-step UnitaryLab circuit and the direct ODE reference. (b) Final-time circuit errors under time-step refinement. (c) Profile refinement with exact finite-Hamiltonian evolution, using $(L,N_p)=(2,64),(4,128),(6,256),(8,256)$. (d) Statevector recovery probabilities before amplitude amplification. Panels (a), (b), and (d) use $r=12$, $M_s=16$, $N_s=N_p=128$, and $L=4$.}
\label{fig:Schrodingerization-validation}
\end{figure}

\FloatBarrier
\subsection{Time-dependent forcing}
\label{subsec:numerical-inhomogeneous}
We retain the same $A(t)$, $\bb x_0=(1,0)^\top$, and $T=1$, and add the source
\begin{equation}
\label{eq:numerical-inhomogeneous-source}
\bb b(t)=0.2\begin{bmatrix}\cos(2\pi t)\\ \sin(2\pi t)\end{bmatrix}.
\end{equation}
Its direction varies with time while $\|\bb b(t)\|=0.2$, giving $g_b=1.2/\|\bb x(1)\|\approx1.62016$. We use $r=8$, $M_s=8$, $N_s=64$, $L=4$, and $N_p=128$, and approximate \eqref{eq:inhomogeneous-grid-duhamel} by Gauss--Legendre quadrature with $m$ nodes. The source on the clock grid is $\bb f_h(t)=\bb G_h(t)\otimes\bb b(t)$, as in \eqref{eq:inhomogeneous-clock-grid}.

We compute independent reference solutions for the physical ODE and the forced clock system using the same Runge--Kutta method. For the projected solution, we evaluate each autonomous propagator by Hermitian eigendecomposition of the retained $p$-mode Hamiltonians. In Fig.~\ref{fig:inhomogeneous-validation}\subref{fig:inhomogeneous-energies}, we show the recovered energies at $t_j=j/16$ using 32 quadrature nodes. Panel~\subref{fig:inhomogeneous-quadrature} separates the final-time total error from the quadrature error, measured against a 128-node projected reference. The quadrature error decreases from $8.85\times10^{-2}$ at $m=2$ to $2.14\times10^{-10}$ at $m=12$, while the total relative vector error levels off near $1.24\times10^{-5}$. The final clock error is $1.01\times10^{-8}$, so the remaining error is predominantly due to $p$-initialization and Fourier projection.

At $T=1$, we also execute the coherent quadrature and recovery circuit in UnitaryLab with $m=16$. Including the initial contribution gives 17 branches, padded to a 32-dimensional quadrature register. The circuit uses 19 qubits: one physical qubit, six clock qubits, seven $p$ qubits, and five quadrature qubits. For this small validation, we supply the preparation oracles and controlled mode propagators as explicit unitary matrices, and implement the $p$-basis conversion by QFT. Thus this circuit tests source combination and joint recovery, without testing the gate synthesis or query cost of the propagator implementation.

The circuit agrees with the projected quadrature vector to $7.08\times10^{-16}$ and gives
\[
\frac{\|\bb x^{\,a}(1)-\bb x(1)\|}{\|\bb x(1)\|}
\approx1.24\times10^{-5},
\qquad
\bigl\|\ket{\bb x^{\,a}(1)}-\ket{\bb x(1)}\bigr\|
\approx1.11\times10^{-5}.
\]
With $B_h\approx2.09840$, the measured statevector probability of jointly selecting the quadrature, $p$, and clock registers is
\[
\mathbb P_{\mathrm{succ}}^{\,a}
=\frac{\|\bb x^{\,a}(1)\|^2}
{B_h^2\|\Pi_p\chi\|_2^2\|\Pi_p\Psi_p\|_2^2}
\approx0.0353613,
\]
in agreement with \eqref{eq:inhomogeneous-success}. No amplitude amplification or finite-shot sampling is used.

\begin{figure}[!htb]
\centering
\subfigure[Recovered dynamics with a time-dependent source.]{%
\includegraphics[width=0.48\textwidth]{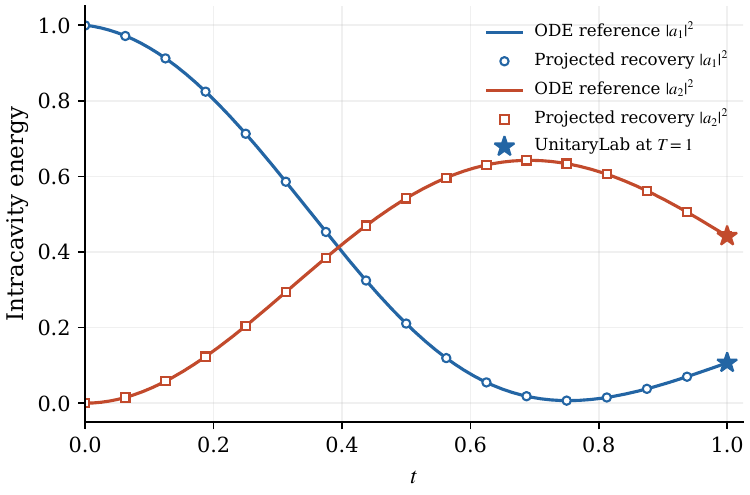}%
\label{fig:inhomogeneous-energies}}
\hfill
\subfigure[Time-quadrature refinement.]{%
\includegraphics[width=0.48\textwidth]{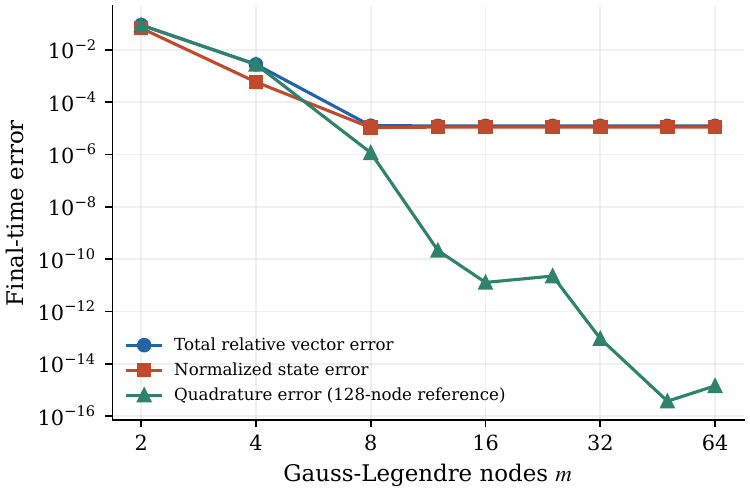}%
\label{fig:inhomogeneous-quadrature}}
\caption{Two-cavity system with the source \eqref{eq:numerical-inhomogeneous-source}. (a) Direct ODE integration and projected recovery with 32 quadrature nodes; stars mark the 16-node UnitaryLab output at $T=1$. (b) Final-time errors under quadrature refinement, with the other discretization parameters fixed. The quadrature reference uses 128 nodes and the same projected propagators.}
\label{fig:inhomogeneous-validation}
\end{figure}

\FloatBarrier
\section{Conclusions}
\label{sec:conclusion}

We have developed an autonomization approach for time-dependent homogeneous dissipative linear systems in which the choice of auxiliary initial profiles controls both approximation accuracy and recovery cost. By combining a compactly supported window function with a normalized Dirichlet kernel, we obtain an autonomous discretization with exponentially decreasing clock error and bounded discrete profile norm. We then apply Schr\"odingerization with an exactly periodic initial profile, using its explicit Fourier coefficients and weighted recovery to simplify the error analysis.

Under the stated finite-order regularity and oracle assumptions, we prepare the normalized final state to accuracy $\varepsilon$ with success probability $\Theta(1)$ using $\mathcal O\!\bigl(g\alpha_AT\log(g\alpha_AT/\varepsilon)\bigr)$ matrix queries and $\mathcal O(g)$ queries to each state-preparation oracle. In particular, we retain a single logarithmic precision factor after accounting for both auxiliary variables. By analyzing recovery jointly, we obtain a success probability $\Theta(g^{-2})$ before amplitude amplification, with no additional normalization penalty that grows with the discretization size. The two-cavity experiments illustrate convergence and recovery probabilities, including coherent source integration for a time-dependent inhomogeneity.

We plan to pursue three directions. First, we will complete the error and complexity analysis for general time-dependent sources in Section~\ref{sec:inhomogeneous}, including source regularity, coherent quadrature, and weight preparation; for constant sources, the augmented construction already reduces the problem to our homogeneous algorithm. By treating the source through Duhamel's formula, we keep its state-preparation queries separate from the matrix queries used in Hamiltonian simulation. Second, we will investigate time-dependent Hamiltonian systems, for which the autonomous generator is already anti-Hermitian and the second auxiliary variable is unnecessary. In this setting, we seek to determine whether a sharper analysis can yield additive dependence on the evolution scale and logarithmic precision. Finally, we will explore applications to time-dependent PDEs and physical models with dissipation, examining how their structure can be incorporated into the autonomization and recovery procedures.

\bibliographystyle{unsrt} 
	\bibliography{Refs}
	
\end{document}